\documentclass[a4paper,UKenglish]{article}

\usepackage{geometry}
\usepackage{amssymb}
\usepackage{amsmath}
\usepackage{amsthm}
\usepackage{xspace}
\usepackage{graphicx}
\usepackage{mdwlist}
\usepackage{color}
\usepackage{soul}
\usepackage{authblk}
\usepackage{algpseudocode}
\usepackage{xparse}
\usepackage{stmaryrd}
\usepackage{units}

\usepackage{changebar}
\let\originalparagraph\paragraph
\renewcommand{\paragraph}[2][.]{\originalparagraph{#2#1}}

\newtheorem{theorem}{Theorem}

\newtheorem{claim}[theorem]{Claim}

\newtheorem{prop}[theorem]{Proposition}

\newtheorem{definition}[theorem]{Definition}
\newtheorem{lemma}[theorem]{Lemma}
\newtheorem{corollary}[theorem]{Corollary}

\newcommand{\E}{\ensuremath{\mathcal{E}}}
\newcommand{\de}{:=}
\newcommand{\set}[1]{\ensuremath{\{{#1}\}}}
\newcommand{\condset}[2]{\ensuremath{\set{{#1}\;|\;{#2}}}}

\newcommand{\tup}[1]{\ensuremath{\langle{#1}\rangle}}

\newcommand{\p}{\mathcal{P}}

\newcommand{\NN}{\ensuremath{\mathbb{N}}\xspace}
\newcommand{\bs}{\ensuremath{\backslash}}
\newcommand{\C}{\ensuremath{\mathcal{C}}\xspace}

\newcommand{\A}{\ensuremath{\mathcal{A}}\xspace}

\newcommand{\B}{\mathbb{B}}
\newcommand{\Bs}{\B_{\mathrm{std}}}
\newcommand{\Bm}{\B_{\mathrm{maj}}}
\newcommand{\Aut}[1]{\mathrm{Aut}(#1)}
\newcommand{\EV}[1]{\mathrm{EV}_{#1}}
\newcommand{\tEV}[1]{\overline{\mathrm{EV}}_{#1}}
\newcommand{\oEV}[1]{\ora{\mathrm{EV}}_{#1}}
\newcommand{\ora}[1]{\vec{#1}}

\newcommand{\npart}[1]{|{#1}|}
\newcommand{\spart}[1]{\|{#1}\|}

\newcommand{\Sym}[1]{\ensuremath{\mathrm{Sym}_{#1}}\xspace}

\newcommand{\sstabu}[1]{\ensuremath{\mathrm{Stab}_U\set{#1}}\xspace}

\newcommand{\stabu}[1]{\ensuremath{\mathrm{Stab}_U(#1)}\xspace}
\newcommand{\stab}[1]{\ensuremath{\mathrm{Stab}(#1)}\xspace}

\newcommand{\sstabn}[1]{\ensuremath{\mathrm{Stab}_n\set{#1}}\xspace}

\newcommand{\ifp}{\ensuremath{\mathrm{ifp}}}

\newcommand{\poly}{\ensuremath{\mathrm{poly}}\xspace}

\newcommand{\FO}{\ensuremath{\mathrm{FO}}\xspace}
\newcommand{\tarb}{\ensuremath{\tau_{\mathrm{arb}}}\xspace}
\newcommand{\CT}{\#}

\newcommand{\PT}{\ensuremath{\mathrm{P}}\xspace}
\newcommand{\Pp}{\ensuremath{\PT/\poly}\xspace}

\newcommand{\NP}{\ensuremath{\mathrm{NP}}\xspace}
\newcommand{\FP}{\ensuremath{\mathrm{FP}}\xspace}
\newcommand{\Lo}{\mathcal{L}}

\newcommand{\LN}{\ensuremath{\Lo+\mathrm{\le}}\xspace}
\newcommand{\LG}{\ensuremath{\Lo+\mathrm{\Advice}}\xspace}
\newcommand{\FPN}{\ensuremath{\mathrm{FP}+\mathrm{\le}}\xspace}

\newcommand{\FPC}{\ensuremath{\mathrm{FPC}}\xspace}

\newcommand{\FPG}{\ensuremath{\mathrm{FP}+\Advice}\xspace}

\newcommand{\FPCG}{\ensuremath{\mathrm{FPC}+\Advice}\xspace}

\newcommand{\CPTC}{\ensuremath{\mathrm{CPTC}}\xspace}

\newcommand{\FON}{\ensuremath{\mathrm{FO}+\mathrm{\le}}\xspace}

\newcommand{\Advice}{\Upsilon}

\newcommand{\Orb}[2]{\ensuremath{\mathrm{Orb}_{#1}({#2})}\xspace}

\newcommand{\sse}{\subseteq}
\newcommand{\ssn}{\subsetneq}
\newcommand{\spe}{\supseteq}

\newcommand{\es}{\emptyset}
\newcommand{\inv}{^{-1}}
\newcommand{\ra}{\rightarrow}

\newcommand{\agr}{\textsc{agree}}
\newcommand{\numf}{\textsc{num}}
\newcommand{\asize}{\textsc{asize}}
\newcommand{\orb}{\mathrm{Orb}}

\newcommand{\supp}{\textsc{supp}}
\newcommand{\spp}{\textsc{overlap}}
\newcommand{\SP}{\mathsf{SP}}
\newcommand{\spt}[1]{{\mathrm{sp}(#1)}}
\newcommand{\ospt}[1]{{\ora{\mathrm{sp}}(#1)}}
\newcommand{\ind}[2]{\ensuremath{[#1:#2]}}
\newcommand{\induce}[1]{\ensuremath{\bar{#1}}}

\newcommand{\str}[1]{\ensuremath{\mathcal{#1}}}
\newcommand{\fin}{\ensuremath{\mathrm{fin}}}

\renewcommand{\bar}[1]{\overline{#1}}

\renewcommand\paragraph[1]{\medskip {\noindent \sffamily\normalsize\bfseries #1}}
\newcommand\iparagraph[1]{\medskip {\noindent \sffamily\normalsize #1}}

\newcommand{\ourtitle}{On Symmetric Circuits and Fixed-Point Logics
\thanks{A short version of this paper is to appear in the proceedings
of STACS 2014.}}

\title{\ourtitle}
\author{Matthew Anderson}
\author{Anuj Dawar}
\affil{University of Cambridge Computer Laboratory\\
15 JJ Thomson Ave, Cambridge, CB3 0FD, UK\\
\texttt{firstname.lastname@cl.cam.ac.uk}}

\begin{document}

\maketitle

\begin{abstract}
We study properties of relational structures such as graphs that are
decided by families of Boolean circuits.  Circuits that decide such
properties are necessarily invariant to permutations of the elements
of the input structures.  We focus on families of circuits that are
symmetric, i.e., circuits whose invariance is witnessed by
automorphisms of the circuit induced by the permutation of the input
structure.  We show that the expressive power of such families is
closely tied to definability in logic.  In particular,
we show that the queries defined on structures by uniform
families of symmetric Boolean circuits with majority gates are exactly
those definable in fixed-point logic with counting.  This shows that
inexpressibility results in the latter logic lead to lower bounds
against polynomial-size families of symmetric circuits.
\end{abstract}

\section{Introduction}
\label{sec:intro}

A property of graphs on $n$ vertices can be seen as a Boolean function
which takes as inputs the $\binom{n}{2}$ potential edges (each of
which can be $0$ or $1$) and outputs either $0$ or $1$.  For the
function to really determine a property of the graph, as opposed to a function of a 
particular presentation of it, the function must be invariant under re-ordering
the vertices of the graph.  That is, permuting the
$\binom{n}{2}$ inputs according to some permutation of $[n]$ leaves 
the value of the function unchanged.  We call such Boolean
functions \emph{invariant}.  Note that this does not require the
Boolean function to be invariant under \emph{all} permutations of its
inputs, which would mean that it was entirely determined by the number
of inputs that are set to $1$.

The interest in invariant functions arises in the context of
characterising the properties of finite 
relational structures (such as finite graphs) that are decidable in
polynomial time.  It is a long-standing open problem in descriptive
complexity to give a characterisation of the polynomial-time
properties of finite relational structures (or, indeed, just graphs)
as the classes of structures definable in some suitable logic (see,
for instance,~\cite[Chapter 11]{EF06}).  It is known that fixed-point
logic $\FP$ and its extension with counting $\FPC$ are strictly less
expressive than deterministic polynomial time P \cite{CFI92}.

It is easy to see that every polynomial-time property of graphs is
decided by a $\PT$-uniform family of polynomial-size circuits that are
\emph{invariant} in the sense above.  On the other hand, when a
property of graphs is expressed in a formal logic, it gives rise to a
family of circuits that is \emph{explicitly invariant} or
\emph{symmetric}.  By this we mean that its invariance is witnessed by
the automorphisms of the circuit itself.  For instance, any sentence
of $\FP$ translates into a polynomial-size family of symmetric Boolean
circuits, while any sentence of $\FPC$ translates into a
polynomial-size family of symmetric Boolean circuits with majority
gates.

Concretely, a circuit $C_n$ consists of a directed acyclic graph whose
internal gates are marked by operations from a basis (e.g., the
standard Boolean basis $\Bs \de \{\text{AND, OR, NOT}\}$ or the majority basis
$\Bm = \Bs \cup \set{\text{MAJ}}$) and input gates which are marked
with pairs of vertices representing potential edges of an $n$-vertex
input graph.  Such a circuit is \emph{symmetric}
if $C_n$ has an automorphism $\pi$ induced by each
permutation $\sigma$ of the $n$ vertices, i.e., $\pi$ moves the input
gates of $C_n$ according to $\sigma$ and preserves operations and
wiring of the internal gates of $C_n$.  Clearly, any
symmetric circuit is invariant.

Are symmetric circuits a weaker model of computation than
invariant circuits?  We aim at characterising the properties that can
be decided by uniform families of symmetric circuits.  Our main result
shows that, indeed, any property that is decided by a uniform
polynomial-size family of symmetric majority circuits can
be expressed in $\FPC$.

\begin{theorem}
  \label{thm:intro-main}
  A graph property is decided by a $\PT$-uniform polynomial-size
  family of symmetric majority circuits if, and only if, it is defined
  by a fixed-point with counting sentence.
\end{theorem}
A consequence of this result is that inexpressibility results that
have been proved for $\FPC$ can be translated into lower bound results
for symmetric circuits.  For instance, it follows (using~\cite{Daw98})
that there is no polynomial-size family of symmetric majority
circuits deciding 3-colourability or Hamiltonicity of graphs.

We also achieve a characterisation similar to Theorem~\ref{thm:intro-main}
of symmetric Boolean circuits.
\begin{theorem}
  \label{thm:intro-bool}
  A graph property is decided by a $\PT$-uniform polynomial-size
  family of symmetric Boolean circuits if, and only if, it is defined
  by a fixed-point sentence interpreted in $\mathcal{G} \oplus
  \tup{[n], \le}$, i.e., the structure that is the disjoint union of
  an $n$-vertex graph $\mathcal{G}$ with a linear order of length
  $n$.
\end{theorem}

Note that symmetric majority circuits can be transformed into
symmetric Boolean circuits.  But, since $\FP$, even interpreted over
$\mathcal{G} \oplus \tup{[n], \le}$, is strictly less expressive
than $\FPC$, our results imply that any such translation must involve
a super-polynomial blow-up in size.  Similarly, our results imply with
\cite{CFI92} that \emph{invariant} Boolean circuits cannot be
transformed into symmetric circuits (even with majority gates) without
a super-polynomial blow-up in size.  On the other hand, it is clear
that symmetric majority circuits can still be translated into
\emph{invariant} Boolean circuits with only a polynomial blow-up.

\paragraph{Support.}
The main technical tool in establishing the translation from uniform
families of symmetric circuits to sentences in fixed-point logics is a
\emph{support theorem} (stated informally below) that establishes
properties of the stabiliser groups of gates in symmetric circuits.

We say that a set $X \subseteq [n]$ \emph{supports} a gate $g$ in a
symmetric circuit $C$ on an $n$-element input structure if every automorphism of
$C$ that is generated by a permutation of $[n]$ fixing $X$ also fixes
$g$.  It is not difficult to see that for a family of symmetric
circuits obtained from a given first-order formula
$\phi$ there is a constant $k$ such that all gates in all circuits of the family have a support of size at most $k$.
To be precise, the gates in such a circuit correspond to subformulas
$\psi$ of $\phi$ along with an assignment of values from $[n]$ to the
free variables of $\psi$.  The set of elements of $[n]$ appearing in
such an assignment forms a support of the gate and its size is
bounded by the number of free variables $\psi$.  Using the fact that
any formula of $\FP$ is equivalent, on structures of size $n$, to a
first-order formula with a constant bound $k$ on the number of
variables and similarly any formula of $\FPC$ is equivalent to a
first-order formula \emph{with majority quantifiers}
(see~\cite{Ott97}) and a constant bound on the number of variables, we
see that the resulting circuits have supports of constant bounded
size.  Our main technical result is that the existence of supports of
bounded size holds, in fact, for all polynomial-size families of
symmetric circuits.  In its general form, we show the following
theorem in Section~\ref{sec:support} via an involved combinatorial
argument.

\begin{theorem}[Informal Support Thm]
  \label{thm:intro-support}
  Let $C$ be a symmetric circuit with $s$ gates over a graph of size
  $n$. If $n$ is sufficiently large and $s$ is sub-exponential in $n$, then
  every gate in $C$ has a support of size $O\hspace{-.5ex}\left(\frac{\log s}{\log n}\right)\hspace{-.5ex}.$
\end{theorem}



In the typical instantiation of the Support Theorem the circuit $C$
contains a polynomial number of gates $s = \poly(n)$ and hence the
theorem implies that every gate has a support that is bounded in size
by a constant.  The proof of the Support Theorem mainly relies on the
structural properties of symmetric circuits and is largely independent
of the semantics of such circuits; this means it may be of independent
interest for other circuit bases and in other settings.

\paragraph{Symmetric Circuits and $\FP$.}
In Section~\ref{sec:sym-def} we show that each polynomial-size 
family $\C$ of symmetric circuits 
can be translated into a formula of
fixed-point logic.  If the family $\C$ is $\PT$-uniform, by the
Immerman-Vardi Theorem \cite{V82,I86} there is an $\FP$-definable
interpretation of the circuit $C_n$ in the ordered structure
$\tup{[n],\le}$.  
We show that the support of a gate is computable in polynomial time, and hence we can also interpret the support of each gate
in $\tup{[n],\le}$.  The circuit $C_n$ can be
evaluated on an input graph $\str{G}$ by fixing a bijection between
$[n]$ and the universe $U$ of $\str{G}$.  We associate with each gate
of $g$ of $C_n$ the set of those bijections that cause $g$ to evaluate to $1$ on $\str{G}$. 
This set of bijections admits a compact
(i.e., polynomial-size) representation as the set of injective maps
from the support of $g$ to $U$.  We show that these compact
representations can be inductively defined by formulas of $\FP$, or
$\FPC$ if the circuit also admits majority gates.

Thus, we obtain that $\PT$-uniform family of symmetric
Boolean circuits can be translated into formulas of $\FP$ interpreted
in $\str{G}$ combined with a disjoint linear order
$\tup{[|\str{G}|],\le}$, while families containing majority gates
can be simulated by sentences of $\FPC$.  The reverse
containment follows using classical techniques. As a consequence we
obtain the equivalences of Theorems~\ref{thm:intro-main}~\&~\ref{thm:intro-bool}, and a number
of more general results as this sequence of arguments naturally
extends to: (i) inputs given as an arbitrary relational structure,
(ii) outputs defining arbitrary relational queries, 
and (iii) non-uniform circuits,
provided the logic is allowed additional advice on the disjoint linear
order.

\paragraph{Related Work.}
We note that the term ``symmetric circuit'' is used by Denenberg et
al. in \cite{DGS86} to mean what we call invariant circuits.  They
give a characterisation of first-order definability in terms of a
restricted invariance condition, namely circuits that are invariant
and whose relativisation to subsets of the universe remains invariant.
Our definition of symmetric circuits follows that in \cite{O97} where
Otto describes it as the ``natural and straightforward combinatorial
condition to guarantee generic or isomorphism-invariant performance.''
He then combines it with a size restriction on the orbits of gates
along with a strong uniformity condition, which he calls ``coherence'', to
give an exact characterisation of definability in infinitary logic.  A
key element in his construction is the proof that if the orbits of
gates in such a circuit are polynomially bounded in size then they
have supports of bounded size.  We remove the assumption of coherence
from this argument and show that constant size supports exist in any
polynomial-size symmetric circuit.  This requires a generalisation of
what Otto calls a ``base'' to supporting partitions.  See
Section~\ref{sec:coherence} for more discussion of connections with
prior work.

\section{Preliminaries}
\label{sec:prelim}

Let $[n]$ denote the set of positive integers $\set{1,\ldots,n}$. 
Let $\Sym{S}$ denote the group of all permutations of the set $S$.
When $S=[n]$, we write $\Sym{n}$ for $\Sym{[n]}$. 

\subsection{Vocabularies, Structures, and Logics}

A \emph{relational vocabulary} (always denoted by $\tau$) is a finite
sequence of relation symbols $(R_1^{r_1},\ldots,R_k^{r_k})$ where for
each $i \in [k]$ the relation symbol $R_i$ has an associated arity
$r_i \in \NN$.  A \emph{$\tau$-structure} $\str{A}$ is a tuple
$\tup{A,R_1^{\str{A}},\ldots,R_k^{\str{A}}}$ consisting of (i) a
non-empty set $A$ called the \emph{universe} of $\str{A}$, and (ii)
relations $R_i^{\str{A}} \sse A^{r_i}$ for $i \in [k]$.  Members of
the universe $A$ are called \emph{elements} of $\str{A}$.  A
\emph{multi-sorted} structure is one whose universe is given as a
disjoint union of several distinct \emph{sorts}.  Define the
\emph{size} of a structure $|\str{A}|$ to be the cardinality of its
universe.  All structures considered in this paper are \emph{finite},
i.e., their universes have finite cardinality.  Let $\fin[\tau]$
denote the set of all finite $\tau$-structures.

\paragraph{First-Order and Fixed-Point Logics.}
Let $\FO(\tau)$ denote \emph{first-order logic} with respect to the
vocabulary $\tau$.  The logic $\FO(\tau)$ is the set of formulas whose
atoms are formed using the relation symbols in $\tau$, an equality
symbol $=$, an infinite sequence of variables $(x,y,z\ldots)$, and
that are closed under the Boolean connectives ($\wedge$ and $\vee$),
negation ($\neg$), and universal and existential quantification
($\forall$ and $\exists$).  Let \emph{fixed-point logic} $\FP(\tau)$
denote the extension of $\FO(\tau)$ to include an inflationary fixed-point
operator $\ifp$.  Assume standard syntax and semantics for \FO and \FP
(see the textbook \cite{EF06} for more background).  For a formula
$\phi$ write $\phi(x)$ to indicate that $x$ is the tuple of the free
variables of $\phi$.  For a logic $\Lo$, a formula $\phi(x) \in
\Lo(\tau)$ with $k$ free variables, $\str{A} \in \fin[\tau]$, and tuple
$a \in A^k$ write $\str{A} \models_{\Lo} \phi[a]$ to express that the
tuple $a$ makes the formula $\phi$ true in the structure $\str{A}$
with respect to the logic $\Lo$.  We usually drop the subscript $\Lo$
and write  $\str{A} \models \phi[a]$ when no confusion would arise.

\paragraph{Logics with Disjoint Advice.}
Let $\tarb$ be a relational vocabulary without a binary relation
symbol $\le$.  Let $\Advice : \NN \rightarrow \fin[\tarb \uplus
  \set{\le^2}]$ be an \emph{advice function}, where for $n \in \NN$,
$\Advice(n)$ has universe $\set{0,1,\ldots,n}$ naturally ordered by
$\le$.  For a logic $\Lo$, typically $\FO$ or $\FP$, let $(\LG)(\tau)$
denote the set of formulas of $\Lo(\tau')$ where $\tau' \de \tau \uplus
\tarb \uplus \set{\le^2}$ and $\tau$ is a vocabulary disjoint from
$\tarb \uplus \set{\le^2}$.  For a structure $\str{A} \in \fin[\tau]$
define the semantics of $\phi \in (\LG)(\tau)$ to be $\str{A}
\models_{(\LG)} \phi \text{ iff } \str{A}^\Advice \models_{\Lo} \phi,$
where $\str{A}^\Advice \de \str{A} \oplus \Advice(|\str{A}|)$ is the
multi-sorted $\tau'$-structure formed by taking the disjoint union of
$\str{A}$ with a structure coding a linear order of corresponding
cardinality endowed with interpretations of the relations in $\tarb$.
The universe of the multi-sorted structure $\str{A}^\Advice$ is written
as $A \uplus \set{0,1,\ldots,|A|}$; refer to $A$ as the \emph{point
  sort} of $\str{A}^\Advice$ 
  and to
$\set{0,1,\ldots,|A|}$ as the \emph{number sort} of $\str{A}^\Advice$.
We are primarily interested in the
special case when $\tarb$ is empty and hence $\Advice(|\str{A}|) =
\tup{\set{0,1,\ldots,|A|},\le}$ is simply a linear order.  Denote
formulas of this logic by $(\LN)(\tau)$ and extended structures by
$\str{A}^{\le}$ to emphasise the disjoint linear order.

Let $\FPC(\tau)$ denote the extension of $(\FPN)(\tau)$ with a
counting operator $\CT_x$ where $x$ is a point or number variable.  For a
structure $\str{A} \in \fin[\tau]$ and a formula $\phi(x) \in
\FPC(\tau)$, $\CT_x \phi(x)$ is a term denoting the element in the number
sort corresponding to $|\condset{a \in \str{A}}{\str{A} \models
  \phi[a]}|$.  See \cite[Section 8.4.2]{EF06} for more details.
Finally, we consider the extension of fixed-point logic with both
advice functions and counting quantifiers $(\FPCG)(\tau)$.

Using $k$-tuples of number variables, it is possible in $\FPN$ and
$\FPC$ to represent numbers up to $n^k$ and perform arithmetic
operations on them.  We omit details but use such constructions
freely. 

\subsection{Symmetric and Uniform Circuits}
\label{sec:boolean}

A \emph{Boolean basis} (always denoted by $\B$) is a finite set of
Boolean functions from $\set{0,1}^*$ to $\set{0,1}$.  We consider only
bases containing symmetric functions, i.e., for all $f \in \B$, $f(x)
= f(y)$ for all $n \in \NN$ and $x,y \in \set{0,1}^n$ with the same
number of ones.  The \emph{standard Boolean basis} $\Bs$ consists of
unbounded fan-in AND, OR, and unary NOT operators.  The \emph{majority
  basis} $\Bm$ extends the standard basis with an operator MAJ which
is one iff the number of ones in the input is at least the number of
zeroes.

\begin{definition}[Circuits on Structures]
A \emph{Boolean $(\B,\tau)$-circuit $C$ with universe $U$ computing a
  $q$-ary query $Q$} is a structure
$\tup{G,W,\Omega,\Sigma,\Lambda}$.
\begin{itemize}
\item $G$ is a set called the \emph{gates} of $C$.  The \emph{size} of $C$ is $|C| \de |G|$.
\item $W \subseteq G \times G$ is a binary relation called the
  \emph{wires} of the circuit.  We require that $(G,W)$ forms a \emph{directed
    acyclic graph}.  Call the gates with no incoming wires
  \emph{input gates}, and all other gates \emph{internal gates}.
  Gates $h$ with $(h,g) \in W$ are called the \emph{children} of
  $g$.
\item $\Omega$ is an injective function from $U^q$ to $G$.  The gates
  in the image of $\Omega$ are called the \emph{output gates}.  When $q
  = 0$, $\Omega$ is a constant function mapping to a single output
  gate.
\item $\Sigma$ is a function from $G$ to $\B \uplus \tau \uplus
  \set{0,1}$ which maps input gates into $\tau \uplus \set{0,1}$ with
  $|\Sigma\inv(0)|,|\Sigma\inv(1)| \le 1$ and internal gates into
  $\B$.  Call the input gates marked with a relation from $\tau$
  \emph{relational gates} and the input gates marked with 0 or 1
  \emph{constant gates}.
\item $\Lambda$ is a sequence of injective functions $(\Lambda_R)_{R \in
    \tau}$ where for each $R \in \tau$, $\Lambda_R$ maps each relational
  gate $g$ with $R = \Sigma(g)$ to $\Lambda_R(g) \in U^r$
  where $r$ is the arity of $R$.  Where no ambiguity arises, we write
  $\Lambda(g)$ for $\Lambda_R(g)$.
\end{itemize}
\end{definition}


Let $C$ be a Boolean $(\B,\tau)$-circuit with universe $U$, $\str{A}
\in \fin[\tau]$ with $|\str{A}| = |U|$, and $\gamma : A \rightarrow U$
be a bijection.  Let $\gamma \str{A}$ denote the $\tau$-structure over
the universe $U$ obtained by relabelling the universe of $\str{A}$
according to $\gamma$.  Recursively evaluate $C$ on $\gamma\str{A}$ by
determining a value $C[\gamma\str{A}](g)$ for each gate $g$: (i) a
constant gate evaluates to the bit given by $\Sigma(g)$, (ii) a
relational gate evaluates to 1 iff $\gamma\str{A} \models
\Sigma(g)(\Lambda(g))$, and (iii) an internal gate evaluates to the
result of applying the Boolean operation $\Sigma(g)$ to the values for
$g$'s children.
$C$ defines the $q$-ary query $Q \sse \str{A}^q$ where $a \in Q$ iff $C[\gamma
  \str{A}](\Omega(\gamma a)) = 1$.

\begin{definition}[Invariant Circuit]
 Let $C$ be a $(\B,\tau)$-circuit with universe $U$ computing a
 $q$-ary query.  The circuit $C$ is \emph{invariant} if for every
 $\str{A} \in \fin[\tau]$ with $|\str{A}| = |U|$, $a \in \str{A}^q$,
 and bijections $\gamma_1,\gamma_2$ from $A$ to $U$,
 $C[\gamma_1\str{A}](\Omega(\gamma_1a)) =
 C[\gamma_2\str{A}](\Omega(\gamma_2a)).$
\end{definition}

Invariance indicates that $C$ computes a property of $\tau$-structures
which is invariant to presentations of the structure.  Moreover, for
an invariant circuit $C$ only the size of $U$ matters and we often
write $C = C_n$ to emphasise that $C_n$ has a universe of size $n$.
A \emph{family} $\C = (C_n)_{n \in \NN}$ of invariant
$(\B,\tau)$-circuits naturally computes a $q$-ary query on
$\tau$-structures.  When $q = 0$ the family computes a Boolean
property of structures.  We now discuss a structural property of
circuits called \emph{symmetry} that implies invariance.

\paragraph{Symmetric Circuits.} 
Permuting a circuit's universe may induce automorphisms of the
circuit.

\begin{definition}[Induced Automorphism]
  Let $C = \tup{G,W,\Omega,\Sigma,\Lambda}$ be a $(\B,\tau)$-circuit
  with universe $U$ computing a
 $q$-ary query.  Let $\sigma \in \Sym{U}$.  If there is a
  bijection $\pi$ from $G$ to $G$ such that
  \begin{itemize}
    \item for all gates $g,h \in G$, $W(g,h)$ iff $W(\pi(g),\pi(h))$, 
    \item for all output tuples $x \in U^q$, $\pi\Omega(x) = \Omega(\sigma(x))$, 
    \item for all gates $g \in G$, $\Sigma(g) = \Sigma(\pi(g))$, and
    \item for each relational gate $g \in G$, $\sigma\Lambda(g) =
      \Lambda(\pi(g))$,
  \end{itemize}
  we say $\sigma$ \emph{induces the automorphism $\pi$} of $C$.
\end{definition}

The principle goal of this paper is to understand the computational
power of circuit classes with the following type of structural
symmetry.

\begin{definition}[Symmetric]\label{def:symmetric}
  A circuit $C$ with universe $U$ is called \emph{symmetric}
  if for every permutation $\sigma \in \Sym{U}$, $\sigma$ induces an
  automorphism of $C$.
\end{definition}
It is not difficult to see that, for a symmetric circuit $C$, there is
a homomorphism $h: \Sym{U} \ra \Aut{C}$ (where $\Aut{C}$ denotes the
automorphism group of $C$) such that $h(\sigma)$ is an automorphism
induced by $\sigma$.  
As long as some element of $U$ appears in the label of some
input gate of $C$, $h$ is an injective homomorphism.  Henceforth we
assume that this is always the case as otherwise $C$ has no relational inputs and computes a constant function.  Circuits where the homomorphism is not also surjective introduce artifacts into our arguments.  To avoid this we require the circuits we consider to be \emph{rigid}.   

\begin{definition}[Rigid]
  \label{def:rigid}
  Let $C = \tup{G,W,\Omega,\Sigma,\Lambda}$ be a $(\B,\tau)$-circuit
  with universe $U$.  Call $C$ \emph{rigid} if there do not exist
  distinct gates $g,g' \in G$ with $\Sigma(g) = \Sigma(g')$,
  $\Lambda(g) = \Lambda(g')$, $\Omega\inv(g) = \Omega\inv(g')$, and
  for every $g'' \in G$, $W(g'',g)$ iff $W(g'',g')$.
\end{definition}

To show that for rigid symmetric circuits $C$, any injective
homomorphism from $\Sym{U}$ to $\Aut{C}$ is surjective, it suffices to
show that each $\sigma \in \Sym{U}$ induces a unique automorphism in
$\Aut{C}$.

\begin{prop}
  \label{prop:rigid}
  Let $C$ be a rigid circuit with universe $U$, and $\sigma \in
  \Sym{U}$.  If $\sigma$ induces an automorphism of $C$, that
  automorphism is unique.
\end{prop}

We defer the proof of this proposition to Section~\ref{subsec:rigid} were we also show that symmetric circuits can be transformed into equivalent rigid symmetric circuits in polynomial time, and hence show that rigidity can be assumed of circuits without loss of generality in our setting.  
For a rigid symmetric circuit $C$, the group of automorphisms of
$C$ is exactly $\Sym{U}$ acting faithfully.  We shall therefore abuse
notation and use these interchangeably.  In particular, we shall write
$\sigma g$ to denote the image of a gate $g$ in $C$ under the action
of the automorphism induced by a permutation $\sigma$ in $\Sym{U}$.

An examination of the definitions suffices to show that symmetry
implies invariance.  In symmetric circuits it is useful to consider
those permutations which induce automorphisms that fix gates.  Let
$\p$ be a partition of a set $U$.  Let the \emph{pointwise stabiliser}
of $\p$ be $\stabu{\p} \de \condset { \sigma \in \Sym{U}}{\forall P
  \in \p, \sigma P = P},$ and similarly define the \emph{setwise
  stabiliser} $\sstabu{\p} \de \condset{ \sigma \in \Sym{U}}{\forall P
  \in \p, \sigma P \in \p}$.  For a gate $g$ in a rigid symmetric circuit $C$ with
universe $U$, let the \emph{stabiliser} of $g$ be $\stabu{g} \de
\condset{ \sigma \in \Sym{U}}{ \induce{\sigma}(g) = g}$, and let the \emph{orbit} of $g$
under the automorphism group $\Aut{C}$ of $C$ be $\orb(g) \de
\condset{\pi g}{\pi \in \Aut{C}}$.
  In each case, when $U=[n]$, we
write $\mathrm{Stab}_n$ instead of $\mathrm{Stab}_{[n]}$.

\paragraph{Uniform Circuits.}
One natural class of circuits are those with
polynomial-size descriptions that can be generated by a deterministic
polynomial-time machine.

\begin{definition}[\PT and \Pp-Uniform]
  A $(\B,\tau)$-circuit family $\C = (C_n)_{n \in \NN}$ computing a
  $q$-ary query is \emph{\Pp-uniform} if there exists an integer $t
  \ge q$ and function $\Advice : \NN \rightarrow \set{0,1}^*$ which
  takes an integer $n$ to a binary string $\Advice(n)$ such that
  $|\Advice(n)| = \poly(n)$, and $\Advice(n)$
  describes\footnote{Formally one must define a particular way of
    encoding circuits via binary strings.  However, since the details
    of the representation are largely irrelevant for our purposes we omit them.}
  the circuit $C_n$ whose gates are indexed by $t$-tuples of
  $\set{0,1,\ldots,n}$, inputs are labelled by $t$-tuples of $[n]$,
  and outputs are labelled by $q$-tuples of $[n]$.  Moreover, if there
  exists a deterministic Turing machine $M$ that for each integer $n$
  computes $\Advice(n)$ from $1^n$ in time $\poly(n)$ call $\C$
  \emph{\PT-uniform}.
\end{definition}

Note that such uniform families implicitly have polynomial size.
It follows from the Immerman-Vardi Theorem \cite{V82,I86} that any
$\PT$-uniform family $\C = (C_n)_{n \in \NN}$ of circuits is definable
by an $\FP$ interpretation in the sense that there is a sequence of
formulas $(\phi_G, \phi_W,
  \phi_\Omega,(\phi_s)_{s \in \B \uplus \tau \uplus \set{0,1}},
  (\phi_{\Lambda_R})_{R\in \tau})$ which, interpreted in the structure $\tup{[n],\le}$,
  describes the circuit $C_n =  \tup{G,W,\Omega,\Sigma,\Lambda}$, with
  
  \begin{itemize} 
  \item $G \sse [n]^t$ such that $g \in G$ iff
    $ \tup{[n],\le}\models \phi_G[g]$.
  \item For all $g,g' \in G$ and 
    $W(g,g')$ iff $\tup{[n],\le} \models \phi_W[g,g']$.
  \item For all $g \in G$ and $a \in [n]^q$, $\Omega(a)
    = g$ iff $\tup{[n],\le} \models \phi_\Omega[a,g]$.
  \item For all $g \in G$ and $s \in \B \uplus \tau \uplus \set{0,1}$,
    $\Sigma(g) = s$ iff $\tup{[n],\le}\models \phi_s[g]$.
  \item For all relational gates $g \in G$ and $a \in
    [n]^r$,
    $\Lambda_R(g) = a$ iff $\tup{[n],\le} \models
    \phi_\Lambda[g,a]$, where $r$ is the arity of $R = \Sigma(g)$.
  \end{itemize}

\noindent More generally, if  $\C = (C_n)_{n \in \NN}$ is a \emph{\Pp-uniform}
family of circuits, there is an $(\FPG)$-definable interpretation of
$C_n$ in $\str{A}^\Advice$ for a suitable advice function $\Advice$.

Over ordered structures 
neither $\PT$-uniform nor $\Pp$-uniform circuits need compute
invariant queries as their computation may
implicitly depend on the order associated with $[n]$.  To obtain
invariance for such circuits we assert symmetry.  The next section
proves a natural property of symmetric circuits that ultimately
implies that \emph{symmetric} \PT-uniform circuits coincide with
$\FP$ definitions on the standard and majority bases.

\section{Symmetry and Support}
\label{sec:support}

In this section we analyse the structural properties of symmetric
circuits. 
We begin with a formal definition of support.
\begin{definition}[Support]\label{def:support}
Let $C$ be a rigid symmetric circuit with universe $U$ and let $g$ be a gate in
$C$. A set $X \subseteq U$ \emph{supports} $g$ if, for any
permutation $\sigma \in \Sym{U}$ such that $\sigma x = x$ for all $x
\in X$, we have $\sigma g=g$ (i.e., $\sigma \in \stabu{g}$).  
\end{definition}



In this section we show how to associate supports of constant size
in a canonical way to all gates in \emph{any} rigid symmetric circuit of
polynomial size.  Indeed, our result is more general as it associates
moderately growing supports to gates in circuits of sub-exponential
size.  As a preliminary to the proof, we introduce, in
Section~\ref{subsec:partitions}, the more general notion of a
\emph{supporting partition} for a permutation group.  We show how to
associate a canonical such supporting partition with any permutation
group $G$ and obtain bounds on the size of such a partition based on
the index of $G$ in the symmetric group.  These results are then used
in, Section~\ref{subsec:suppthm}, to bound the size of partitions supporting stabiliser groups of gates  based
on the size of the circuit, proving our main technical result---the
Support Theorem.

\subsection{Supporting Partitions}\label{subsec:partitions}

The notion of a supporting partition generalises the notion of a
support of a gate by replacing the set with a partition and
the stabiliser group of the gate with an arbitrary permutation group.

\begin{definition}[Supporting Partition]
Let $G \sse \Sym{U}$ be a group and $\p$ a partition of $U$.  We say
that $\p$ is a \emph{supporting partition} of $G$ if $\stabu{\p} \sse
G$.  
\end{definition}

For intuition consider two extremes.  When $G$ has supporting
partition $\p = \set{U}$, it indicates $G =\Sym{U}$.  Saying that $G$
has supporting partition $\p =
\set{\set{u_1},\set{u_2},\ldots,\set{u_{|U|}}}$ indicates only that
$G$ contains the identity permutation, which is always true.

A natural partial order on partitions is the coarseness relation,
i.e., $\p'$ is as coarse as $\p$, denoted $\p' \spe \p$, if every part in $\p$ is
contained in some part of $\p'$.  For two partitions $\p$ and $\p'$,
there is a most refined partition that is as coarse as either partition:
\begin{definition}
  \label{def:E}
  Let $\p,\p'$ be partitions of $U$.  Define a binary relation
  $\sim$ on $U$: $u_1 \sim u_2$ iff there exists $P \in \p \cup \p'$
  such that $u_1,u_2 \in P$.  Let $\E(\p,\p')$ denote the
  partition of $U$ corresponding to the equivalence classes of $U$
  under the transitive closure of $\sim$.
\end{definition}
Now it is easy to show that $\E$ preserves supporting partitions (the proof
is similar to that of (*) on page 379 of \cite{O97}).
\begin{prop}
  \label{prop:intersect_gen}
Let $G \sse \Sym{U}$ be a group and $\p,\p'$ be supporting partitions
of $G$.  Then  $\E(\p,\p')$ is also a supporting partition of $G$.
\end{prop}
\begin{proof}
  Let $\E \de \E(\p,\p') = \set{E_1,\ldots,E_m}$.  Suppose $\sigma \in
  \stab{\E}$ and we now show that $\sigma \in G$.   Because the parts $E_i$ are
  disjoint write $\sigma$ as $\sigma_1\cdots \sigma_m$ where $\sigma_i
  \in \Sym{E_i}$ (i.e., it permutes only the elements of $E_i$).
  Indeed each $\sigma_i$ may be written as a sequence of
  transpositions of elements in $E_i$.  Thus it suffices to show that
  each transposition $(u u')$ with $u,u' \in E_i$ can be written as a
  sequence of permutations in $\stab{\p} \cup \stab{\p'} \sse
  G$.  Since $u,u' \in E_i$ there is a sequence of elements of $u_1,
  \ldots, u_\ell$ with $u_1 = u$, $u_\ell = u'$ and $u_j \sim u_{j+1}$
  for $j \in [\ell-1]$ witnessing the path from $u$ to $u'$.  By the
  definition of $\sim$, for each $j \in [\ell-1]$ there is $P \in \p
  \cup \p'$ such that $u_j, u_{j+1} \in P$ and therefore
  $(u_ju_{j+1})$ is a transposition in $\stab{\p} \cup \stab{\p'}$.
  Conclude that the transposition $(uu') = (u_1u_\ell) =
  (u_1u_2)(u_2u_3)\cdots
  (u_{\ell-2}u_{\ell-1})(u_{\ell-1}u_\ell)(u_{\ell-2}u_{\ell-1})\cdots
  (u_1u_2)$ is a sequence of transpositions from $\stab{\p} \cup
  \stab{\p'}$ and the proof is complete.
\end{proof}
This implies that each permutation group has a unique coarsest
partition that supports it.
\begin{lemma} \label{lem:supp-gen} 
Each permutation group $G\sse \Sym{U}$ has a unique coarsest supporting partition.
\end{lemma}
\begin{proof}
  Suppose $G$ has two distinct coarsest
  partitions $\p,\p'$ of the universe $U$ that support it, then
  Proposition~\ref{prop:intersect_gen} implies that the coarser
  partition $\E(\p,\p')$ also supports $G$.  This is a contradiction.
\end{proof}

We write $\SP(G)$ for \emph{the unique coarsest partition supporting} $G$.
For a partition $\p$ of $U$ and a permutation $\sigma \in \Sym{U}$,
we write $\sigma\p$ for the partition $\condset{\sigma P}{P \in \p}$.
Note that this commutes with the operation $\E$, so $\sigma\E(\p,\p')
= \E(\sigma\p,\sigma\p')$.
The next lemma shows how supporting partitions are affected by the
conjugacy action of $\Sym{U}$.
\begin{lemma}\label{lem:supp-conjug1}
If $\p$ is a partition supporting a group $G$, then for any $\sigma
\in \Sym{U}$, $\sigma\p$ supports the group $\sigma G \sigma^{-1}$.
\end{lemma}

\begin{proof}
Let $\pi \in \stabu{\sigma \p}$
and let $P$ be a part in $\p$, then:
$$ (\sigma^{-1}\pi\sigma)P = (\sigma^{-1}\pi)(\sigma P) =
\sigma^{-1}\sigma P = P,$$
where the second equality follows from the fact that $\pi$ fixes
$\sigma P$.  Thus, $\sigma^{-1}\pi\sigma$ fixes $\p$ pointwise,
therefore $\sigma^{-1}\pi\sigma \in G$ and hence $\pi \in \sigma G
\sigma^{-1}$.
\end{proof}
This indicates how the unique coarsest supporting partition of a group translates under conjugation.
\begin{lemma}\label{lem:supp-conjug}
For any $G \sse \Sym{U}$ and any $\sigma \in \Sym{U}$, $\sigma \SP(G) = \SP(\sigma G \sigma^{-1})$.
\end{lemma}
\begin{proof}
Immediate from Lemma~\ref{lem:supp-conjug1} and the fact that the
action of $\E$ commutes with $\sigma$.
\end{proof}
We conclude that any group $G$ is sandwiched between the pointwise and setwise
stabilisers of $\SP(G)$.
\begin{lemma}\label{lem:supp-group}
For any group $G \sse \Sym{U}$, we have $\stabu{\SP(G)} \sse G \sse
\sstabu{\SP(G)}$. 
\end{lemma}
\begin{proof}
The first inclusion is by definition of supporting partitions.  For
the second, note that if $\sigma \in G$, then $\sigma G \sigma^{-1} =
G$.  Then, by Lemma~\ref{lem:supp-conjug}, $\sigma\SP(G) = \SP(G)$.
\end{proof}   
Note that these bounds need not be tight.
For example, if $G$ is the alternating group on $U$
(or, indeed, any transitive, primitive subgroup of $\Sym{U}$),
then $\SP(G)$ is the partition of $U$ into singletons.  In this case,
$\stabu{\SP(G)}$ is the trivial group while $\sstabu{\SP(G)}$ is all
of $\Sym{U}$. 

We now use the bounds given by Lemma~\ref{lem:supp-group},
in conjunction with bounds on $G$ to obtain size bounds on $\SP(G)$.
Recall that the index of $G$ in $\Sym{U}$, denoted
$\ind{\Sym{U}}{G}$ is the number of cosets of $G$ in $\Sym{U}$ or,
equivalently, $\frac{|\Sym{U}|}{|G|}$.
The next lemma says that if $\p$ is a partition of $[n]$ where the
index of $\sstabn{\p}$ in $\Sym{n}$ is sufficiently small then the
number of parts in $\p$ is either very small or very big.
\begin{lemma}\label{lem:small-large}
For any $\epsilon$ and $n$ such that $0\leq \epsilon < 1$ and $\log n \geq
\frac{4}{\epsilon}$, if $\p$ is a partition of $[n]$ with $k$ parts,
$s \de \ind{\Sym{n}}{\sstabn{\p}}$ and $n \leq s \leq
2^{n^{1-\epsilon}}$, then $\min\set{k,n-k} \le \frac{8}{\epsilon}
\frac{\log s}{\log n}$.
\end{lemma}
\begin{proof}
 Let $p_1 \le p_2 \le \ldots \le p_k$ be the
  respective sizes of the parts in $\p$.  Thus,
  \begin{equation} 
    \label{eq:weak-k:1}
    \hfill s =
    \frac{n!}{|\sstabn{\p}|} \ge \frac{1}{k!}\frac{n!}{p_1! p_2!
      \cdots p_k!}.\hfill
  \end{equation}
Observe that,  if $p_i > 1$, then $p_1! p_2! \cdots p_k! \leq p_1! p_2!
\cdots (p_i -1)! \cdots (p_k+1)!$.  By repeatedly applying this, we
see that in the lower bound on $s$ given by
Equation~\eqref{eq:weak-k:1}, we can replace $p_1! p_2! \cdots p_k!$
by $(n - (k-1))!$.
Let $k' \de
  \min\set{k,n-k}$ and we have
  $$s \ge \frac{n!}{k!
  (n-(k-1))!} = \frac{1}{n+1}\binom{n+1}{k} \ge
  \frac{1}{n+1}\binom{n}{k} \ge \frac{1}{n+1}\binom{n}{k'} \ge
  \frac{1}{n+1}\left(\frac{n}{k'}\right)^{k'}$$ where the second
  inequality follows because $\binom{n}{k} = \binom{n}{n-k} =
  \binom{n}{k'}$, and the third inequality follows from a simple
  combinatorial bound.  Take the logarithm of both sides of the above
  equation to get $\log s \ge k' (\log n - \log k') - \log (n+1)$.
  Using the fact that $s \ge n \ge 2$ ($\log n \ge \frac{4}{\epsilon}
  \ge 4$) then implies that
  \begin{equation}
    \hfill 4\log s \ge k' (\log n - \log k'), \label{eq:weak-k:2}\hfill
  \end{equation}
  The definition of $k'$ implies that $k' \le \frac{n}{2}$ and $\log n
  - \log k' \ge 1$.  Plugging this into Equation~\eqref{eq:weak-k:2}
  gives that $4\log s \ge k'.$ Take the logarithm of this inequality
  and apply the upper bound on $s$ to determine that $(1-\epsilon)\log
  n + 2 \ge \log k'.$ Inserting this inequality back into
  Equation~\eqref{eq:weak-k:2} gives $4 \log s \ge k' (\epsilon \log n
  - 2).$ Since $\frac{\epsilon}{2}\log n \ge 2$, conclude that $k' \le
  \frac{8}{\epsilon}\frac{\log s}{\log n}$.
\end{proof}

We use a similar argument to establish that, under the assumptions of
the previous lemma, when the number of parts in $\p$ is small, then
the largest part is very big. 

\begin{lemma}\label{lem:largepart}
For any $\epsilon$ and $n$ such that $0\leq \epsilon < 1$ and $\log n
\geq \frac{8}{\epsilon^2}$, if $\p$ is a partition of $[n]$ with
$\npart{\p} \leq \frac{n}{2}$, $s \de \ind{\Sym{n}}{\sstabn{\p}}$
and $n \leq s \leq 2^{n^{1-\epsilon}}$, then $\p$ contains a part $P$
with at least $n - \frac{33}{\epsilon} \cdot \frac{\log s}{\log n}$
elements. 
\end{lemma}
\begin{proof}
  The initial setup is the same as in the proof of
  Lemma~\ref{lem:small-large}.   Let $p_1 \le p_2
  \le \ldots \le p_k$ be the respective sizes of the parts in $\p$ and
  let $S \de \sum_{i = 1}^{k-1}p_i$.  Our aim is to show that $S \leq
  \frac{33}{\epsilon} \cdot \frac{\log s}{\log n}$. 
  Denote the size of the second largest part by $p \de p_{k-1}$.  We
  have
  \begin{equation} 
    \label{eq:rk:1}
    \hfill s = \frac{n!}{|\sstabn{\p}|} \ge \frac{1}{k!}\frac{n!}{p_1!p_2!
      \cdots p_k!}.\hfill
  \end{equation}
  Let $\ell \in \NN$ be such that
  \begin{equation}
    \label{eq:rk:2}
    \hfill \ell \le k-1, \text{ and } k-1 + \ell (p-1) \le \sum_{i =
      1}^{k-1}p_i = S. \hfill
  \end{equation}
   Provided $\p$ contains more than one part both $\ell \in \set{0,1}$
  satisfy Equation~\eqref{eq:rk:2}.   We may assume that $p > 1$
  otherwise $S \le \npart{\p}$ and we are done by
  Lemma~\ref{lem:small-large}.
 For any $\ell \ge 1$ satisfying
  Equation~\eqref{eq:rk:2}, redistributing weight from a $p_i$ to $p_j$
  with $i < j$ in a way similar to the proof of
  Lemma~\ref{lem:small-large} gives the following,
  \begin{align}
    s &\ge \frac{1}{k!}\frac{n!}{ \underbrace{1!\cdots1!}_{k-1-\ell
        \text{~times}} \underbrace{p!\cdots p!}_{\ell
        \text{~times}}(n-(k-1-\ell + \ell p))!} \ge \frac{n!}{
      (p!)^\ell (n-\ell(p-1) + 1)!} \notag \\
    &\ge \frac{(\frac{n}{e})^n}{(e \sqrt{p}(\frac{p}{e})^p)^\ell (e
      \sqrt{n}
      (\frac{n}{e})^{n-\ell(p-1)+1})}=\frac{n^{\ell(p-1)-3/2}}{p^{\ell(p+1/2)}}
    \underbrace{\frac{e^{\ell p} e^{n-\ell(p-1)+1}}{e^{\ell+1}
        e^n}}_{=1} \notag
  \end{align}
  where the third inequality follows from Stirling's Formula, i.e.,
  that for any $x \ge 2$, $(\frac{x}{e})^x \le x! \le \sqrt{2\pi
    x}(\frac{x}{e})^x e^{\frac{1}{12x}} \le e \sqrt{x}
  (\frac{x}{e})^x$.  Take the logarithm of the above equation to
  determine that
  \begin{align}
    \label{eq:rk:3}
    \log s &\ge \left[\ell(p-1)-\frac{3}{2}\right] \log n - \ell \left(p
    +\frac{1}{2}\right) \log p \\ &= \ell \left(p + \frac{1}{2}\right) (\log n
    - \log p) - \frac{3}{2}(\ell-1)\log n, \notag \\
    \label{eq:rk:4}
    \frac{5}{2} \log s &\ge \ell p (\log n - \log p) -
    \frac{3}{2}\ell \log n \ge p (\log n - \log p) -
    \frac{3}{2}\log n, \\
    \label{eq:rk:5}
    4\log s &\ge  p (\log n - \log p) \ge p, 
  \end{align}
  where Equation~\eqref{eq:rk:4} follows from Equation~\eqref{eq:rk:3}
  since $s \ge n$ and $\ell \ge 1$, and Equation~\eqref{eq:rk:5}
  follows from Equation~\eqref{eq:rk:4} because $p$ is the size of the
  second largest part of $\p$ and hence $p \le \frac{n}{2}$ and $(\log
  n - \log p) \ge 1$.  Take the logarithm of Equation~\ref{eq:rk:5}
  and use the bound on $s$ to determine that $\log p \le \log \log s +
  2 \le (1-\epsilon) \log n + 2$.  Plug this bound into
  Equation~\eqref{eq:rk:4} to get that $\frac{5}{2} \log s \ge
  \ell p (\epsilon \log n - 2) - \frac{3}{2} \ell \log n$.  Using
  $\frac{\epsilon}{2} \log n \ge 2$ and dividing by $\log n$,
  $\frac{5}{2}\frac{\log s}{\log n} \ge \ell (\frac{p\epsilon}{2} -
  \frac{3}{2})$.

    If $\frac{p \epsilon}{4} \ge \frac{3}{2}$, then $\frac{10}{\epsilon}
  \frac{\log s}{\log n} \ge \ell p$.  For the largest value of $\ell$,
  $k-1 + (\ell+1)(p-1) \ge S$, and hence $k-1 + 2\ell p
  \ge S$.  Thus  Lemma~\ref{lem:small-large} implies that
  $S \le \frac{8 + 20}{\epsilon}\frac{\log s}{\log n}$.
  Otherwise $p < \frac{6}{\epsilon}$ and hence $\log p \le 3 - \log
  \epsilon$.  Plugging this into Equation~\eqref{eq:rk:4} and using
  $\log n \ge \frac{8}{\epsilon^2} \ge 2(3 - \log \epsilon) \ge 2\log
  p$ gives $\frac{5}{2} \log s \ge \ell \frac{p-3}{2} \log n.$ If $p
  \ge 5$, then recover $\frac{25}{2\epsilon}\frac{\log s}{\log n} \ge
  \ell p$ and conclude $S \le \frac{8 +
    25}{\epsilon}\frac{\log s}{\log n}$ analogously to before.  Otherwise
  $p \le 4$, and $S \le p (k-1) \le 4 \cdot
  \frac{8}{\epsilon}\frac{\log s}{\log n}$ by
   Lemma~\ref{lem:small-large}.  Since each case concluded that
  $S \le \frac{33}{\epsilon}\frac{\log s}{\log n}$ the
  proof is complete.
\end{proof}

\subsection{Support Theorem}\label{subsec:suppthm}

Here we leverage the two combinatorial lemmas of the last subsection 
to show that in symmetric circuits of polynomial size, each gate has a
small supporting partition, and hence has a small support.

Let $g$ be a gate in a symmetric circuit $C$ over universe $U$, from
now on, we abuse notation and write $\SP(g)$ for $\SP(\stabu{g})$.
Note that, if $P$ is any part in $\SP(g)$, then $U\setminus P$ is a
support of $g$ in the sense of Definition~\ref{def:support}.  We write
$\spart{\SP(g)}$ to denote the smallest value of $|U\setminus P|$ over
all parts $P$ in $\SP(g)$.  Also, let $\SP(C)$ denote the maximum of
$\spart{\SP(g)}$ over all gates $g$ in $C$.

By the orbit-stabiliser theorem, $|\orb(g)| =
\ind{\Sym{U}}{\stabu{g}}$.  By Lemma~\ref{lem:supp-group}, we have
that $\stabu{g} \sse \sstabu{\SP(g)}$ and thus, if $s$ is an upper
bound on $|\orb(g)|$, $s \ge \ind{\Sym{U}}{\stabu{g}} \ge
\ind{\Sym{U}}{\sstabu{\SP(g)}}$.  Then, by Lemma~\ref{lem:largepart},
$g$ has a support of small size provided that (i) $s$ is
sub-exponential, and (ii) $\SP(g)$ has fewer than $n/2$ parts.  Thus,
to prove our main technical theorem, which formalises
Theorem~\ref{thm:intro-support} from the introduction, it suffices to
show that if $s$ is sufficiently sub-exponential, (ii) holds.


\begin{theorem}[Support Theorem]
  \label{thm:support} 
  For any $\epsilon$ and $n$ with $\frac{2}{3} \le \epsilon \le 1$ and
  $n > 2^{\frac{56}{\epsilon^2}}$, if $C$ is a rigid symmetric circuit over
  universe $U$ with $|U| = n$ and $s \de \max_{g \in C} |\orb(g)| \le
  2^{n^{1-\epsilon}}$, then, $\SP(C) \le \frac{33}{\epsilon}
  \frac{\log s}{\log n}.$
\end{theorem}
\begin{proof}
  Suppose $1 \le s < n$.  $C$ cannot have relational inputs, because
  each relational gate must have an orbit of size at least $n$, so
  each gate of $C$ computes a constant Boolean function.  The support
  of every gate $g$ in $C$ must be $\set{U}$, and hence $0 =
  \spart{\SP(g)} = \SP(C)$.  Therefore assume $s \ge n$.  
  
  To conclude the theorem from Lemma~\ref{lem:largepart} it
  suffices to argue that for all gates $g$, $\npart{\SP(g)} \le
  \frac{n}{2}$.  Suppose $g$ is a constant gate, then, because $g$ is
  the only gate with its label, it is fixed under all permutations and
  hence $\npart{\SP(g)} = \npart{\set{U}} = 1 < 
  \frac{n}{2}$.  If $g$ is a relational gate, then it is fixed by any
  permutation that fixes all elements appearing in $\Lambda(g)$ and moved by
  all others.  Thus, $\SP(g)$ must contain singleton parts for each
  element of $U$ in 
  $\Lambda(g)$ and a part containing everything else.  Thus, if
  $\npart{\SP(g)} > \frac{n}{2}$, $\SP(g)$ contains at least
  $\frac{n}{2}$ singleton parts, there is a contradiction using the
  bounds on $s, n,$ and $\epsilon$, $s \ge |\orb(g)| \ge
  \spart{\SP(g)}! \cdot \binom{n}{\spart{\SP(g)}} \ge
  \left\lfloor\frac{n}{2}\right\rfloor!  \ge
  2^{\lfloor\frac{n}{4}\rfloor} > 2^{n^{1-\epsilon}}.$

  It remains to consider internal gates. For the sake of contradiction
  let $g$ be a topologically first internal gate such that $\SP(g)$ has more
  than $\frac{n}{2}$ parts.  Lemma~\ref{lem:small-large} implies,
  along with the assumptions on $s, n,$ and $\epsilon$, that
  $n-\npart{\SP(g)} \le k' \de \left\lceil\frac{8}{\epsilon}\frac{\log
    s}{\log n}\right\rceil \le \frac{1}{4} n^{1 - \epsilon} <
  \frac{n}{2}.$

  Let $H$ denote the children of $g$.  Because $g$ is a topologically
  first gate with $\npart{\SP(g)} > \frac{n}{2}$, for all $h \in H$,
  $\SP(h)$ has at most $\frac{n}{2}$ parts.  As before, we
  argue a contradiction with the upper bound on $s$.  This done by
  demonstrating that there is a set of gate-automorphism pairs $S =
  \condset{(h,\sigma)}{h \in H, \sigma \in \Sym{U}}$ that are: (i) \emph{useful} -- the automorphism
  moves the gate out of the set of $g$'s children, i.e., $\sigma h
  \not\in H$, and (ii) \emph{independent} -- each child and its image
  under the automorphism are fixed points of the other automorphisms
  in the set, i.e., for all $(h,\sigma), (h', \sigma') \in S$, $\sigma' h =
  h$ and $\sigma' \sigma h = \sigma h$.  Note that sets which are useful and
  independent contain tuples whose gate and automorphism parts are all
  distinct.  The set $S$ describes elements in the orbit of $H$ with
  respect to $\Sym{U}$.

  \begin{claim}
    \label{claim:p}
    Let $S$ be useful and independent, then $|\Orb{}{H}| \ge
    2^{|S|}$.
  \end{claim}
  \begin{proof}
    Let $R$ be any subset of $S$.  Derive an automorphism from $R$:
    $\sigma_R \de \prod_{(h,\sigma) \in R} \sigma$ (since automorphisms
    need not commute fix an arbitrary ordering of $S$).

    Let $R$ and $Q$ be distinct subsets of $S$ where without loss of
    generality $|R| \ge |Q|$.  Pick any $(h,\sigma) \in R \bs Q \neq
    \es$.  Because $S$ is independent $\sigma_R h = \sigma h$ and $\sigma_Q\sigma
    h = \sigma h$.  Since $S$ is useful, $\sigma h \not\in H$.  Thus $\sigma h
    \in \sigma_R H$, but $\sigma h \not\in \sigma_Q H$.  Hence $\sigma_R H \neq
    \sigma_Q H$.  Therefore each subset of $S$ can be identified with a distinct element
    in $\orb(H)$ and hence $|\orb(H)| \ge 2^{|S|}$.
  \end{proof}

  Thus to reach a contradiction it suffices to construct a
  sufficiently large set $S$ of gate-automorphism pairs.  To this end, divide $U$
  into $\lfloor\frac{|U|}{k'+2}\rfloor$ disjoint sets $S_i$ of size
  $k'+2$ and ignore the elements left over.  Observe that for each $i$
  there is a permutation $\sigma_i$ which fixes $U \bs S_i$ but
  $\sigma_i$ moves $g$, because otherwise the supporting partition of $g$
  could be smaller ($n - (k'+2)+1$).  Since $g$ is moved by each
  $\sigma_i$ and $C$ is rigid, there must be an associated child $h_i \in H$
  with $\sigma_i h_i \not\in H$.  Thus let
  $(h_i,\sigma_i)$ be the gate-automorphism pair for $S_i$,
  these pairs are \emph{useful}.  Let $Q_i$ be the union of all but
  the largest part of $\SP(h_i)$.  Observe that for any $\sigma$ which
  fixes $Q_i$ pointwise $\sigma$ also fixes both $h_i$ and
  $\sigma_i h_i$, simply by the definition of support.

  Define a directed graph $K$ on the sets $S_i$ as follows.  Include
  an edge from $S_i$ to $S_j$, with $i \neq j$, if $Q_i \cap S_j \neq
  \es$.  An edge in $K$ indicates a potential lack of independence
  between $(h_i,\sigma_i)$ and $(h_j,\sigma_j)$, and on the other hand if
  there are no edges between $S_i$ and $S_j$, the associated pairs are
  independent.  Thus it remains to argue that $K$ has a large
  independent set.  This is possible because the out-degree of $S_i$
  in $K$ is bounded by $|Q_i| = \spart{\SP(h_i)} \le
  \frac{33}{\epsilon}\frac{\log s}{\log n}$ as the sets $S_i$
  are disjoint and Lemma~\ref{lem:largepart} can be
  applied to $h_i$.  Thus the average total degree (in-degree +
  out-degree) of $K$ is at most $9k'$.  Greedily select a maximal
  independent set in $K$ by repeatedly selecting the $S_i$ with the
  lowest total degree and eliminating it and its neighbours.  This
  action does not effect the bound on the average total degree of $K$
  and hence determines an independent set $I$ in $K$ of size at least
  $$\frac{\lfloor\frac{|U|}{k'+2}\rfloor}{9 k'+1} \ge
  \frac{n-(k'+2)}{(9k'+1)(k'+2)} \ge \frac{\frac{n}{2}-1}{9k'^2+10k'+2}
  \ge \frac{\frac{7}{16}n}{9k'^2+10k'+2} \ge \frac{n}{(7 k')^2}$$ where
  the first inequality follows by expanding the floored expression,
  the second follows because $k' < \frac{n}{2}$, the third follows
  from the lower bound on $n$, and the last follows because $k' \ge 1$
  as it is the ceiling of a positive non-zero quantity by definition.
 
  Take $S \de \condset{(h_i,\sigma_i)}{S_i \in I}$.  By the argument
  above $S$ is useful and independent.  By Claim~\ref{claim:p},
  conclude that $s \ge |\orb(g)| \ge |\orb(H)| \ge 2^{|S|} \ge
  2^{\frac{n}{(7k')^2}}.$ For $\epsilon \ge \frac{2}{3}$, $s \le
  2^{n^{1-\epsilon}}$, and $\frac{\epsilon}{56} \log n > 1$ the
  following is a contradiction
  $\log s \ge n \cdot (\frac{56}{\epsilon}\frac{\log s}{\log
        n})^{-2} > n \cdot (n^{1-\epsilon})^{-2} =
  n^{2\epsilon-1} \ge n^{1-\epsilon}.$
  Thus $\npart{\SP(g)}
  \le \frac{n}{2}$ for all $g \in C$ and the proof is complete by
  Lemma~\ref{lem:largepart}. 
\end{proof}

Observe that when $s$ is polynomial in $n$ the support of a rigid symmetric
circuit family is asymptotically constant.  This is the case for
polynomial-size families.

\begin{corollary}
  \label{cor:support}
  Let $\C$ be a polynomial-size rigid symmetric circuit family, then
  $\SP(\mathcal{C}) = O(1)$.
\end{corollary}

\section{Translating Symmetric Circuits to Formulas}
\label{sec:sym-def}

In this section, we deploy the support theorem to show that $\PT$-uniform
families of symmetric circuits can be translated into
formulas of fixed-point logic.  As a first step, we argue in
Section~\ref{subsec:rigid} that we can restrict our attention to rigid circuits, by
showing that every symmetric circuit can be
converted, in polynomial time, into an equivalent rigid symmetric
circuit.  In Section~\ref{subsec:supports} we show that there are
polynomial-time algorithms that will determine whether a circuit is
symmetric and, if so, compute for every gate its coarsest supporting
partition and therefore its canonical support.  In
Section~\ref{subsec:evaluate} we give an inductive construction of
a relation that associates to each gate $g$ of $C$ a set of tuples
that when assigned to the support of $g$ result in $g$ being evaluated
to true.  This construction is turned into a definition in fixed-point
logic in Section~\ref{subsec:formulas}.

\subsection{Rigid Circuits}\label{subsec:rigid}

We first argue that rigid circuits uniquely induce automorphisms.

\begin{proof}[Proof of Proposition~\ref{prop:rigid}]
  Let $\sigma \in \Sym{U}$ induce the automorphisms $\pi, \pi'$ of
  $C$.  We show $\pi g = \pi' g$ for all gates $g$ in $C$, and hence
  $\pi = \pi'$.

  Observe that if $g$ is an output gate, the image of $g$ under any
  automorphism induced by $\sigma$ must be
  $\Omega(\sigma\Omega\inv(g))$, because $\Omega$ is a function, and
  hence $\pi g = \pi' g$ is unique and completely determined by $\sigma$.
  Therefore assume that $g$ is not an output gate.  We proceed by
  induction on the height of $g$ to show that $\pi g = \pi' g$.

  In the base case $g$ is an input gate.  If $g$ is a constant gate,
  $g$ is the only constant gate of its type and hence all
  automorphisms of $C$ must fix it.  If $g$ is a relational gate, $g$
  is the only relational gate with its type $\Sigma(g)$ and label
  $\Lambda(g)$ and it must map to the similarly unique gate with type
  $\Sigma(g)$ and tuple $\sigma \Lambda(g)$ and hence $\pi g = \pi
  g'$.

  In the induction step $g$ is an internal gate.  By rigidity of $C$,
  $g$ is unique for its children and type.  Moreover, by induction the
  children of $g$ map in the same way under $\pi$ and $\pi'$, and
  hence the image of $g$ must be the same in both automorphisms.  Thus
  $\pi g = \pi' g$ for all gates of $C$.
\end{proof}

To see that any symmetric circuit can be transformed in polynomial
time into an equivalent rigid symmetric circuit, observe that we can proceed
inductively from the input gates, identifying gates whenever they
have the same label and the same set of children.  This allows us to
establish the following lemma.
\begin{lemma}
  \label{lem:alg-rigid}
  Let $C = \tup{G,W,\Omega,\Sigma,\Lambda}$ be a $(\B,\tau)$-circuit
  with universe $U$.  There is a
  deterministic algorithm which runs in time $\poly(|C|)$ and outputs
  a rigid $(\B,\tau)$-circuit $C'$ with gates $G' = G$ such that for
  any $g \in G$, any input $\tau$-structure $\str{A}$ and any
  bijection $\gamma$ from $A$ to $U$, $C[\gamma\str{A}](g) =
  C'[\gamma\str{A}](g)$. 
Moreover, $C'$ is symmetric if $C$ is.
\end{lemma}
  \begin{proof}
  Partition the gates $G$ into equivalence classes where gates in the
  same class have the same labels, output markings, and children.  If
  $C$ is rigid every class has size one, otherwise there is at least
  one class containing two gates.  

  Let $E$ be a minimum height equivalence class containing at least two gates. 
  Order the gates in $E$: $g_1,g_2,\ldots,g_{|E|}$.   For each gate $f \in G \bs E$, let $c_f$ denote the number of wires from $E$ to $f$, and note that $c_f \le |E|$.   For all gates in $E$ remove all outgoing wires.   For all gates $E \bs \set{g_1}$: (i) remove all input wires, and (ii) set their operation to AND.  For each $i$, $1 \le i \le |E|-1$, add a wire from $g_i$ to $g_{i+1}$.  For each $f \in G \bs E$ and $i \in [|E|]$, add a wire from $g_i$ to $f$ if $c_f \le i$.  This completes the transformation of the gates in $E$.
  
  We now argue that this does not effect the result computed at any
  gate $g$.
  First observe that no output gates appear in $E$,  because
  $\Omega$ is injective and hence each output gate must be the sole
  member of its equivalence class.  All gates in $E$ originally had
  identical sets of children and labels and hence they must have evaluated to the
  same value.  The modifications made do not change this property as
  $g_1$ computes the value it originally would have, then passes this
  value to the other gates in $E$, along a chain of single input AND
  gates.  The modifications to the outgoing wires of $E$ insure that
  each gate that originally took input from $E$ has the same number of
  inputs from $E$ (each with the same value) in the modified circuit.
  Taken together this means that the result computed at any gate in
  the modified circuit is the same as that computed at that gate in $C$.
  
  We next argue that the local modification of $E$ makes strict progress towards producing a rigid circuit $C'$.  The local modification of $E$ can only change equivalence classes above $E$ because
the changes to the output wires of $E$ are the only thing that can possibly effect other equivalence classes. After the modification all gates in $E$ must be in singleton equivalence classes because each gate in $E$ is designed to have a unique set of children.  

  Greedily applying the above local modification simultaneously to all topologically minimal non-singleton equivalence classes of $C$, until none remain, produces a rigid circuit $C'$ that computes the same query as $C$, because, as we have just argued, equivalence classes cannot grow as a result of this local modification.  Moreover, this must happen after at most $|C|$ many local modifications, because the number of equivalence classes is at most $|C|$.
 
  We now show that this transformation preserves symmetry.   Suppose $C$ is symmetric.   
  Fix any permutation $\sigma \in \Sym{U}$.   Let $\pi$ be an automorphism induced by $\sigma$ on $C$.  Observe that any induced automorphism on $C$ must map equivalence classes
  to equivalence classes because labels and children are preserved.   
  It is easy to translate $\pi$ into an induced automorphism of $C'$.
  Let $E$ and $E'$ be two equivaluence classes such that $\pi E = E'$
  where $g_1,\ldots,g_{|E|}$ and $g_1',\ldots,g_{|E'|}'$ are the
  ordering of the gates in $E$ and $E'$ in $C'$.  It can be argued by
  induction that mapping $g_i$ to $g_i'$ for all $1 \le i \le
  |E|=|E'|$ preserves all labels and wires and hence is an induced
  automorphism of $\sigma$ in $C'$.  Since $\sigma$ is arbitrary, we
  conclude that the resulting circuit is symmetric.

  The construction of equivalence classes and, indeed, the overall
  construction of $C'$ can be easily implemented in time polynomial in
  $|C|$ when given the circuit in a reasonable binary encoding.  Finally, as gates are only being rewired and relabelled, $G = G'$.
\end{proof}

\subsection{Computing Supports}\label{subsec:supports}

By Lemma~\ref{lem:alg-rigid}, we know that there is a polynomial-time
algorithm that converts a circuit into an equivalent rigid circuit
while preserving symmetry.  In this subsection we show how to, in
polynomial time, check whether the resulting circuit is symmetric, and
if it is, compute the support of each gate.  To this end we first describe an algorithm for determining induced automorphisms of a rigid circuit.

\begin{lemma} \label{lem:alg-action}
  Let $C$ be a rigid $(\B,\tau)$-circuit with universe $U$ and $\sigma
  \in \Sym{U}$.  There is a deterministic algorithm which runs in time
  $\poly(|C|)$ and outputs for each gate $g \in G$ its image under the
  automorphism $\pi$ induced by $\sigma$, if it exists.
\end{lemma}

\begin{proof}
  Process the gates of $C$ recursively building up a mapping $\pi$.
  Compute the mapping for the children of a gate $g$ before
  determining the mapping for $g$.  If at any point an image for $g$
  cannot be located, halt and output that there is no induced
  automorphism.

  Let $g$ be a constant gate, then $g$ is fixed under every
  automorphism.  Let $g$ be a relational gate, then there is at most
  one gate $g'$ in $C$ with $\Sigma(g) = \Sigma(g')$, $\sigma
  \Lambda(g) = \Lambda(g')$, and $\sigma \Omega\inv(g) =
  \Omega\inv(g')$.  If $g'$ exists, set $\pi g$ to $g'$, otherwise
  halt with failure.  Similarly, when $g$ is an internal gate use
  $\Lambda$, $\Omega$, and the action of $\pi$ on the children of $G$
  (via $W$) to determine a unique image of $g$, if it exists.

  By Proposition~\ref{prop:rigid} if $\sigma$ induces an automorphism
  of $C$, it is unique and will be discovered by the above algorithm.
  This algorithm clearly runs in time polynomial in $|C|$.
\end{proof}

Using the preceding lemma we can determine whether a given rigid
circuit is symmetric by computing the set of automorphisms induced by
transpositions of the universe.  If an induced automorphism fails to
exist the circuit cannot be symmetric.  Otherwise, it must be
symmetric because such transpositions generate the symmetric 
group.  If the circuit is symmetric, the coarsest supporting
partitions and orbits of each gate can be determined by examining the
transitive closure of the action of the automorphisms induced by
transpositions on the universe and the gates, respectively.

\begin{lemma}
  \label{lem:alg-sym}
  Let $C$ be a rigid $(\B,\tau)$-circuit with universe $U$.  There is a
  deterministic algorithm which runs in time $\poly(|C|)$ and decides
  whether $C$ is symmetric.  If $C$ is symmetric the algorithm also
  outputs the orbits and coarsest supporting partitions of every gate.
\end{lemma}

\begin{proof}
  For all transpositions $(uv) \in \Sym{U}$ run the algorithm of
  Lemma~\ref{lem:alg-action} to determine the unique automorphism
  $\pi_{(uv)}$ of $C$ induced by $(uv)$, if it exists.  Output that
  $C$ is symmetric iff every induced automorphism $\pi_{(uv)}$ exists.
  This is correct because the set of transpositions generates all of
  $\Sym{U}$, and therefore the automorphisms $\pi_{(uv)}$ generate all induced automorphisms of $C$.  

  If $C$ is symmetric, these induced automorphisms also indicate the
  supporting partitions and orbits of each gate $g$.  Let $\p_{(uv)} \de
  \set{\set{u,v}} \cup_{w \in U \bs \set{u,v}} \set{\set{w}}$ be a
  partition of $U$.  Note that $\pi_{(uv)}$ fixes $g$ iff $\p_{(uv)}$
  supports $g$.  Let $\p$ be the partition determined by combining the
  partitions $\p_{(uv)}$ which support $g$ using $\mathcal{E}$.
  Proposition~\ref{prop:intersect_gen} implies that $\p$ supports $g$.
  Suppose $\p$ is not the coarsest partition supporting $g$.  Then,
  there exists $u,v \in U$ which are not in the same part of $\p$ but
  in the same part of some partition supporting $g$.  But by
  the definition of $\p$, $\pi_{(uv)}$ cannot fix $g$---a contradiction.  Therefore $\p$ is the coarsest partition
  supporting $g$.

  To compute the orbit of a gate $g$: Start with $S_0 \de \set{g}$,
  and for $i \ge 0$, compute $S_{i+1} \de S_i \cup_{(uv) \in \Sym{U}}
  \pi_{(uv)} S_i$.  Let $S$ be the least fixed point of this process.  We argue that $S = \orb(g)$.  $S \sse \orb(g)$, because it consists
  of gates reachable from $g$ via a sequence of induced automorphisms
  of $C$.  $S \spe \orb(g)$, because the set of automorphisms induced
  by transpositions generate the the group of all induced
  automorphisms.

  Since there are only $\binom{|U|}{2}$ transpositions, and we can
  determine whether there is an induced automorphism for each
  transposition in time $\poly(|C|)$ and hence determine whether $C$
  is symmetric in time $\poly(|C|)$.  If $C$ is symmetric the
  computation of the supports and orbits of all gates also is computed
  in time $\poly(|C|)$ because each output is completely determined by
  the equivalence classes induced by the relations defined by the
  induced automorphisms $\pi_{(uv)}$.  Therefore the overall algorithm
  runs in time $\poly(|C|)$.
\end{proof}

\subsection{Succinctly Evaluating Symmetric Circuits}
\label{subsec:evaluate}

Let $\C = (C_n)_{n \in \NN}$ be a family of polynomial-size rigid symmetric
circuits computing a $q$-ary query.  Let $n_0$ be a constant sufficient to apply the
Support Theorem to $C_n$ for $n \ge n_0$ and fix such an $n$.  By
Theorem~\ref{thm:support}, there is a constant bound $k$ so that for
each gate $g$ in $C_n$ the union of all but the largest part of the
coarsest partition supporting $g$, $\SP(g)$, has at most $k$ elements.
Moreover, this union is a
\emph{support} of $g$ in the sense of Definition~\ref{def:support}.
We call it the \emph{canonical support} of $g$ and denote it by
$\spt{g}$.  In this subsection we show that how a gate $g$ evaluates in
$C_n$ with respect to a structure $\str{A}$ depends only on how the universe $U$ of
the structure is mapped to the canonical support of $g$.  This allows
us to succinctly encode the bijections which make a gate true (first
as injective partial functions and then as tuples).  This ultimately
lets us build a fixed-point formula for evaluating $C_n$---indeed, all
symmetric circuits---in the next subsection.

For any set $X \sse [n]$, let $U^X$ denote the set of \emph{injective}
functions from $X$ to $U$.  Let $X, Y \sse [n]$ and $\alpha \in
U^X, \beta \in U^Y$, we say $\alpha$ and $\beta$
are \emph{consistent}, denoted $\alpha \sim \beta$, if for all $z \in
X \cap Y, \alpha(z) = \beta(z)$, and for all $x \in X \bs Y$ and
$y \in Y \bs X$, $\alpha(x) \neq \beta(y)$.  Recall that any bijection
$\gamma: U \ra [n]$ determines an evaluation of the circuit $C_n$ on
the input structure $\str{A}$ which assigns to each gate $g$ the
Boolean value $C_n[\gamma\str{A}](g)$.  (Note that $\gamma\inv \in
U^{[n]}$.)  Let $g$ be a gate and let
$\Gamma(g) \de \condset{\gamma}{C_n[\gamma\str{A}](g) = 1}$ denote the
set of those bijections which make $g$ evaluate to $1$.  The following
claim proves that the membership of $\gamma$ in $\Gamma(g)$ (moreover,
the number of 1s input to $g$) depends only on what $\gamma$ maps to
$\spt{g}$.

\begin{claim}\label{claim.gamma}
Let $g$ be a gate in
$C_n$ with children $H$.  Let $\alpha \in U^\spt{g}$, then for all
$\gamma_1, \gamma_2 : U \rightarrow [n]$ with
$\gamma_1\inv \sim \alpha$ and $\gamma_2\inv \sim \alpha$,
\begin{enumerate}
\item $\gamma_1 \in \Gamma(g)$ iff $\gamma_2 \in \Gamma(g)$. \label{claim.gamma.i}
\item $|\condset{ h \in H }{\gamma_1 \in \Gamma(h)}| = |\condset{ h \in H
}{\gamma_2 \in \Gamma(h)}|.$ \label{claim.gamma.ii}
\end{enumerate}
\end{claim}

\begin{proof}
There is a unique permutation $\pi \in \Sym{n}$ such that $\gamma_1
= \pi\gamma_2$.  Moreover, $\pi$ fixes $\spt{g}$ pointwise, since
$\gamma_1\inv$ and $\gamma_2\inv$ are consistent with $\alpha$.  Since
$C_n$ is rigid and symmetric, $\pi$ is an automorphism of $C_n$, and
we have that $C_n[\gamma_1\str{A}](g) = C_n[(\pi\gamma_1)\str{A}](\pi
g)$.  Since $\pi$ fixes $\spt{g}$ pointwise, we have $\pi g = g$ and
therefore $C_n[\gamma_1\str{A}](g) = C_n[(\pi\gamma_1)\str{A}](g) =
C_n[\gamma_2\str{A}](g)$, proving part~\ref{claim.gamma.i}.
Similarly, for any child $h \in H$ we have that
$C_n[\gamma_1\str{A}](h) = C_n[(\pi\gamma_1)\str{A}](\pi h) =
C_n[\gamma_2\str{A}](\pi h)$.  Since $\pi$ fixes $g$, $\pi$
fixes $H$ setwise.  As this establishes a bijection between the children
$H$ that evaluate to 1 for $\gamma_1$ and $\gamma_2$, we conclude
part~\ref{claim.gamma.ii}.
\end{proof}

We associate with each gate $g$ a set of injective functions
$\EV{g} \sse U^{\spt{g}}$ defined by $\EV{g} \de \condset{\alpha \in
U^\spt{g}}{\exists \gamma \in \Gamma(g) \wedge \alpha \sim \gamma\inv}$
and note that, by Claim~\ref{claim.gamma}, this completely determines
$\Gamma(g)$.  We can use the following claim to recursively construct
$\EV{g}$ for all gates in $C$.

\begin{claim}\label{claim.ev}
Let $g$ be a gate in
$C$ with children $H$.  Let $\alpha \in U^\spt{g}$, then for all
$\gamma : U \rightarrow [n]$ with $\gamma\inv \sim \alpha$,
\begin{equation}
\hfill
|\condset{ h \in H }{\gamma \in \Gamma(h)}| = \sum_{h \in
H} \frac{|A_h \cap \EV{h}|}{|A_h|}, \label{eqn.ev} \hfill
\end{equation}
where for $h \in H$, $A_h \de \condset{\beta \in
U^{\spt{h}}}{\alpha \sim \beta}$.
\end{claim}

\begin{proof}
We have,
\begin{equation} 
\begin{aligned}
|\condset{ h \in H }{\gamma \in \Gamma(h)}| &\cdot |\condset{\delta \in
 U^{[n]}}{\delta \sim \alpha}| \\
&= \sum_{\condset{\delta \in
 U^{[n]}}{ \delta \sim \alpha}} |\condset{h \in
 H}{\delta\inv \in \Gamma(h)}|\\
&= \sum_{h \in H} \sum_{\condset{\delta \in
U^{[n]}}{ \delta \sim \alpha}} |\set{\delta\inv \in \Gamma(h)}| \\
&= \sum_{h \in H} \sum_{\beta \in A_h} \sum_{\condset{\delta \in
U^{[n]}}{ \delta \sim \alpha \wedge \delta \sim \beta}}
|\set{\delta\inv \in \Gamma(h)}| \\
&= \sum_{h \in H} \sum_{\beta \in A_h} |\set{\beta \in \EV{h}}| \cdot
|\condset{\delta \in
U^{[n]}}{ \delta \sim \alpha \wedge \delta \sim \beta}| \\
\end{aligned}
\end{equation}
where the first equality follows from Claim~\ref{claim.gamma}
Part~\ref{claim.gamma.ii}, the second by linearity of addition (note
that $|\set{\delta\inv \in \Gamma(h)}| \in \set{0,1}$), the third by
the definitions of $\sim$ and $A_h$, and the fourth by the definition of
$\EV{h}$.  Observing that $|\condset{\delta \in
U^{[n]}}{ \delta \sim \alpha \wedge \delta \sim \beta}| /
|\condset{\delta \in U^{[n]}}{ \delta \sim \alpha}| = \frac{1}{|A_h|}$,
we conclude that
$$|\condset{ h \in H }{\gamma \in \Gamma(h)}| = \sum_{h \in
H} \sum_{\condset{\beta \in
U^\spt{h}}{\beta \sim \alpha}} \frac{|\set{\beta \in \EV{h}}|}{|A_h|}
= \sum_{h\in H} \frac{|A_h \cap \EV{h}|}{|A_h|}.$$
\end{proof}
Note that implicit in the claim is that the r.h.s.\ side of \eqref{eqn.ev}
is integral.

Since $[n]$ is linearly ordered, $X \sse [n]$ inherits this order and
we write $\vec{X}$ for the ordered $|X|$-tuple consisting of the elements of
$X$ in the inherited order.  For $\alpha \in U^X$ we write
$\vec{\alpha} \in U^{\vec{X}}$ to indicate the tuple
$\alpha(\vec{X})$.  Observe that this transformation is invertible.
This allows us to succinctly encode such injective functions as tuples
over $U$ and, further, to write relational analogs of the sets of
injective functions we considered before, e.g.,
$\oEV{g} \de \condset{\vec{\alpha}}{\alpha \in \EV{g}}.$ Using
Claim~\ref{claim.ev} is it easy to recursively define $\oEV{g}$ over
$C_n$.

\begin{itemize}
\item 
Let $g$ be a constant input gate, then $\spt{g}$ is empty.  If
$\Sigma(g) = 0$, then $\Gamma(g) = \emptyset$ and $\oEV{g}
= \emptyset$.  Otherwise $\Sigma(g) = 1$, then $\Gamma(g)$ is all
bijections and $\oEV{g} = \set{\tup{}}$, i.e., the set containing
the empty tuple.
\item 
Let $g$ be a relational gate with $\Sigma(g) = R \in \tau$, then
$\spt{g}$ is the set of elements in the tuple $\Lambda_R(g)$.  By
definition we have $\oEV{g} = \condset{\vec{\alpha} \in
U^{\ospt{g}}}{\alpha (\Lambda_R(g)) \in R^\A}$.
\item 
Let $\Sigma(g) = \text{AND}$ and consider $\vec{\alpha} \in
 U^\ospt{g}$.  By Claim~\ref{claim.ev},
 $\vec{\alpha} \in \oEV{g}$ iff $\vec{A_h} = \oEV{h}$
 for \emph{every} child $h$ of $g$, i.e., for \emph{every} child $h$
 and \emph{every} $\beta \in U^\spt{h}$ with $\alpha \sim \beta$, we
 have $\vec{\beta} \in \oEV{h}$.
\item 
Let $\Sigma(g) = \text{OR}$ and consider $\vec{\alpha} \in
U^\ospt{g}$.  By Claim~\ref{claim.ev}, $\vec{\alpha} \in \oEV{g}$
iff there is a child $h$ of $g$ where $\vec{A_h} \cap \oEV{h}$ is
non-empty, i.e., for \emph{some} child $h$ of $g$ and \emph{some}
$\beta \in U^\spt{h}$ with $\alpha \sim \beta$, we have
$\vec{\beta} \in \oEV{h}$.
\item 
Let $\Sigma(g) = \text{NOT}$ and consider $\vec{\alpha} \in
U^\ospt{g}$.  The gate $g$ has exactly one child $h$.
Claim~\ref{claim.ev} implies that $\vec{\alpha} \in \oEV{g}$ iff
$\vec{A_h} \neq \oEV{h}$, i.e., for \emph{some} $\beta \in
U^\spt{h}$ with $\alpha \sim \beta$, we have
$\vec{\beta} \not\in \oEV{h}$.
\item 
Let $\Sigma(g) = \text{MAJ}$ and consider $\vec{\alpha} \in
U^\ospt{g}$.  Let $H$ be the set of children of $g$ and let
$A_h \de \condset{ \beta \in U^{\spt{h}}}{\beta\sim\alpha}$.  Then
Claim~\ref{claim.ev} implies that $\vec{\alpha} \in \oEV{g}$ if,
and only if,
\begin{equation}\label{eqn:maj-gate}
\hfill
\sum_{h \in H} 
\frac{| \vec{A_h} \cap \oEV{h} |}{|\vec{A_h}|} \ge \frac{|H|}{2}.
\hfill
\end{equation}
\end{itemize}
From $\oEV{}$ we can recover the $q$-ary query $Q$ computed by
$C_n$ on the input structure $\str{A}$ because the support of an output gate $g$ is exactly
the set of elements in the marking of $g$ by
$\Lambda_\Omega$.  In particular:
 $$Q = \condset{\bar a \in U^q}{\exists g \in
G, \vec\alpha \in \oEV{g}\text{ such that
} \Lambda_\Omega(\alpha\inv(\bar a)) = g}.$$ For Boolean properties $q
= 0$, and $Q = \set{\tup{}}$ indicates that $\str{A}$ has the property and
$Q = \emptyset$ indicates that it does not.

\subsection{Translating to Formulas of $\FP$}\label{subsec:formulas}

Let $\C = (C_n)_{n \in \NN}$ be a $\PT$-uniform family of symmetric
$(\B,\tau)$ circuits, where $\B$ is either $\Bs$ or $\Bm$.  Our aim is
to show that there is a formula $Q$ of $\FP$, or $\FPC$ in the
case of $\Bm$, in the vocabulary $\tau\uplus \set{\le}$ such that for
any $n$ and $\tau$-structure $\str{A}$ over a universe $U$ with
$|U|=n$, the $q$-ary query defined by $C_n$ on input $\str{A}$ is
defined by the formula $Q$ when interpreted in the structure
$\str{A}^\le \de \str{A} \uplus \tup{[n],\le}$.

Since $\C$ is $\PT$-uniform, by the Immerman-Vardi theorem and Lemma~\ref{lem:alg-rigid}, we have an
$\FP$ interpretation defining a rigid symmetric circuit equivalent to $C_n$---that we also call $C_n$---over the number sort of
$\str{A}^\le$, i.e., a sequence $\Phi \de
(\phi_G, \phi_W, \phi_\Omega,(\phi_s)_{s \in \B \uplus \tau \uplus \set{0,1}},
(\phi_{\Lambda_R})_{R\in \tau})$
of formulas of $\FP(\le)$ that define the circuit when interpreted in
$\tup{[n],\le}$.  Note that $C_n$ is defined over the universe
$[n]$.  Let $t$ be the arity of the interpretation, i.e., $\phi_G$
defines a $t$-ary relation $G \sse [n]^t$.  If $n$ is less than
$n_0$, the length threshold for applying the support theorem, $C_n$
can be trivially be evaluated by a $\FP$ formula which quantifies over
all (constantly-many) bijections from the point sort of $\str{A}^\le$
to the number sort of $\str{A}^\le$ and then directly evaluates the
circuit with respect to the bijection.  Thus we only need to consider
the case when $n \ge n_0$, and are able to use the recursive
construction of $\oEV{}$ from the last subsection along with a
constant bound $k$ on the size of the gate supports in $C_n$.

A small technical difficulty arises from the fact that we want to
define the relation $\oEV{g}$ inductively, but these are actually
relations of varying arities, depending on the size of $\spt{g}$.  For
the sake of a uniform definition, we extend $\oEV{g}$ to a $k$-ary
relation for all $g$ by padding it with all possible values to obtain
tuples of length $k$.  If $|\spt{g}| = \ell$, define $\tEV{g} =
\condset{(a_1\cdots a_k)}{(a_1\cdots a_\ell) \in \oEV{g} \mbox{ and } a_i
  \neq a_j \mbox{ for } i \neq j}$.

Define the relation $V \sse [n]^t \times U^k$ by $V(g,\bar a)$ if, and
only if, $\bar a \in \tEV{g}$.  Our aim is to show that the relation $V$ is
definable by a formula of $\FP$.  Throughout this subsection we use $\mu$ and $\nu$ to indicate $t$-tuples of number variables which denote gate indexes in $[n]^t$, and use the $k$-tuples of point variables $\bar x = (x_1,\ldots,x_k)$ and $\bar y = (y_1,\ldots, y_k)$ to denote injective functions that have been turned into tuples and then padded.

By Lemma~\ref{lem:alg-sym} and invoking the Immerman-Vardi theorem again,
we have a formula $\supp$ such that $\tup{[n],\le} \models \supp[g,u]$ if,
and only if, $\tup{[n],\le} \models \phi_G[g]$ (i.e., $g$ is a gate of $C_n$
as defined by the interpretation $\Phi$) and $u$ is in $\spt{g}$.
We use $\supp$ to define some additional auxiliary formulas.
First we define, for each $i$ with $1 \leq i \leq k$ a
formula $\supp_i$ such that  $\tup{[n],\le} \models \supp_i[g,u]$ if, and
only if, $u$ is the $i^{th}$ element of $\ospt{g}$.
These formulas can be defined inductively as follows, where $\eta$ is a number variable
$$
\begin{array}{rcl}
\supp_1(\mu,\eta) & \de & \supp(\mu,\eta) \land \forall \chi (\chi < \eta) \ra \neg\supp(\mu,\chi) \\
\supp_{i+1}(\mu,\eta) & \de &  \supp(\mu,\eta) \land \exists \chi_1
(\chi_1 < \eta \land \supp_i(\mu,\chi_1) \\
& &~~~~~~~~~~~~~~~~~~~~~~~\land \forall \chi_2 [(\chi_1
< \chi_2 < \eta) \ra \neg\supp(\mu,\chi_2)]) .
\end{array}
$$

We now define a formula $\agr(\mu,\nu,\bar{x},\bar{y})$
so that for a structure $\str{A}$,
$\str{A}^{\leq} \models \agr[g,h,\bar{a},\bar{b}]$ if, and only if,
$\alpha \sim \beta$ for $\vec\alpha \in U^\ospt{g}$, $\vec\beta \in U^\ospt{h}$ that
are the restrictions of the
$k$-tuples $\bar{a}$ and $\bar{b}$ to the length of $\ospt{g}$ and
$\ospt{h}$ respectively.
$$
\begin{array}{rcl@{}l}
\agr(\mu,\nu,\bar{x},\bar{y}) & \de & 
\displaystyle \bigwedge_{1\leq i,j, \leq k} & (\forall \eta
(\supp_i(\mu,\eta) \land \supp_j(\nu,\eta)) \ra x_i = y_j) \land \\
& & & (\forall \eta_1 \eta_2 (\supp_i(\mu,\eta_1) \land \supp_j(\nu,\eta_2) \land
x_i = y_j) \ra \eta_1 = \eta_2)
\end{array}
$$

With these, we now define a series of formulas
$(\theta_s)_{s \in \B \uplus \tau \uplus \set{0,1}}(\mu,x)$
corresponding to the various cases of the construction of the relation
$\oEV{g}$ from Section~\ref{subsec:evaluate}.  In these, $V$ is a relational variable for the relation
being inductively defined.
$$
\begin{array}{rcl@{}l}
\theta_0(\mu,\bar{x}) & \de & \text{false} \\
\theta_1(\mu,\bar{x}) & \de & \displaystyle \bigwedge_{1\leq i<j \leq k}
x_i \neq x_j \\
\theta_R(\mu,\bar{x}) & \de & \displaystyle  \bigwedge_{1\leq i<j \leq k}
x_i \neq x_j \land  \exists z_1\cdots z_r \exists \eta_1\cdots \eta_r 
R(z_1,\ldots,z_r) \land \phi_{\Lambda_R}(\mu,\eta) \land \\
& & \hspace{23ex} \displaystyle \bigwedge_{1\leq i \leq r
} \bigwedge_{1\leq j \leq k } (\supp_j(\mu,\eta_i) \ra z_i = x_j) \\
\theta_{\text{OR}}(\mu,\bar{x}) & \de & \displaystyle  \bigwedge_{1\leq i<j \leq k}
x_i \neq x_j \land \exists \nu \exists \bar{y}
(W(\nu,\mu) \land \agr(\mu,\nu,\bar{x},\bar{y})) \land V(\nu,\bar{y})) \\
\theta_{\text{AND}}(\mu,\bar{x}) & \de & \displaystyle  \bigwedge_{1\leq i<j \leq k}
x_i \neq x_j \land \forall \nu \forall \bar{y}
(W(\nu,\mu) \land \agr(\mu,\nu,\bar{x},\bar{y})) \ra V(\nu,\bar{y})) \\
\theta_{\text{NOT}}(\mu,\bar{x}) & \de & \displaystyle  \bigwedge_{1\leq i<j \leq k}
x_i \neq x_j \land \forall \nu \forall \bar{y}
(W(\nu,\mu) \land \agr(\mu,\nu,\bar{x},\bar{y})) \ra \neg V(\nu,\bar{y}))
\end{array}
$$

To define $\theta_{\text{MAJ}}$ we start with some observations.  We
wish to formalise Equation~\ref{eqn:maj-gate}, but there are a few
complications.  The relation $k$-ary $\tEV{h}$ we are defining
inductively is the result of padding $\oEV{h}$ with all tuples of
$k-|\spt{h}|$ distinct elements.  Thus the number of elements in
$\tEV{h}$ is $|\oEV{h}|\cdot\frac{(n-|\spt{h}|)!}{(n-k)!}$.  Similarly,
for any fixed $g$, $h$ and $\bar a$, if we write $\bar{A}_h$ for the set of
tuples $\bar b$ satisfying $\agr(g,h,\bar a,\bar b)$, then $|\bar{A}_h| =
|\vec A_h|\cdot\frac{(n-|\spt{h}|)!}{(n-k)!}$.  Finally, the tuples in
$\bar{A}_h \cap \tEV{h}$ are exactly those obtained by padding tuples
in $\vec A_h \cap \oEV{h}$ to length $k$ and there are therefore
$|\vec A_h \cap \oEV{h}|\cdot\frac{(n-|\spt{h}|)!}{(n-k)!}$ many of these.
Thus, 
$\frac{|\vec A_h \cap \oEV{h}|}{|\vec A_h|} = \frac{| \bar{A}_h \cap \tEV{h}|}{|\bar{A}_h|}$ and it suffices to compute the latter.
Observe that $|\vec A_h|$ and $|\bar{A}_h|$ are completely determined by $|\spt{g}|$, $|\spt{h}|$ and
$|\spt{g} \cap \spt{h}|$. 
We avoid dealing explicitly with fractions by noting that for any gate
$h$, the sum $\sum_{h' \in \orb(h)} 
\frac{| \vec{A}_{h'} \cap \oEV{h'} |}{|\vec{A}_{h'}|}$ is an integer
(by an argument analogous to Claim~\ref{claim.ev}).  Since
$|\bar{A}_{h'}|$ is the same for all $h' \in \orb(h)$, it suffices to
compute the sum of $| \bar{A}_{h'} \cap \tEV{h'} |$ for all $h'$ with a fixed
size of  $|\bar{A}_h|$ and then divide the sum by  $|\bar{A}_h|$.
This is what we use to compute the sum on the l.h.s.\ of
Equation~\ref{eqn:maj-gate}.   

For any fixed $i$ and $j$ with $0 \leq i \leq  j \leq k$, define the
formula $\spp_{ij}(\mu,\nu)$ 
so that
$\str{A}^{\leq} \models \spp_{ij}[g,h]$ iff $|\spt{h}| = j$ and
$|\spt{g}\cap\spt{h}| = i$.  This formula can be defined in $\FO$.

Using $k$-tuples of number variables in $\FPC$ we can represent
natural numbers less than $n^k$.  We assume, without giving detailed construction
of the formulas involved, that we can define arithmetic operations on
these numbers.  In particular, we assume we have for each $i,j$ as
above a formula
$\asize_{ij}(\mu,\xi)$, with $\xi$ a
$k$-tuple of number variables, such that
$\str{A}^{\leq} \models \asize_{ij}[g,e]$ iff $e = |\bar{A}_h|$ for
any gate $h$ with  $|\spt{h}| = j$ and $|\spt{g}\cap\spt{h}| = i$.

Using this, we define the formula $\numf_{ij}(\mu,\bar{x},\xi)$, with $\xi$ a $k$-tuple of
number variables,
so that
$\str{A}^{\leq} \models \numf_{ij}[g,\bar a,e]$ iff $e$ is the
number of gates $h$ with $\str{A}^{\leq} \models \spp_{ij}[g,h]$ which
are made true by some bijection 
that assigns the tuple $\bar a$ to
$\spt{g}$.  This formula is given by
$$
\begin{array}{rcll}
\numf_{ij}(\mu,\bar{x},\xi) & \de &
\exists \xi_1 \xi_2 & \asize_{ij}(\mu,\xi_1) \land \\
& & & \xi_2 = \CT_{\nu \bar y}(W(\nu,\mu) \land \spp_{ij}(\mu,\nu) \land
V(\nu,\bar{y}) \land \agr(\mu,\nu,\bar{x},\bar{y}) ) \land \\
& & & \xi \cdot \xi_1 = \xi_2.
\end{array}
$$
Now we can define the required formula $\theta_{\text{MAJ}}$ by 
$$
\theta_{\text{MAJ}}(\mu,\bar{x}) \de \displaystyle  \bigwedge_{1\leq i<j \leq k}
x_i \neq x_j \land \exists \xi (2\cdot \xi \geq
\CT_{\nu} W(\nu,\mu)) \land \xi = \sum_{0\leq
i \leq j \leq k} \condset{\xi'}{\numf_{ij}(\mu,\bar{x},\xi')},
$$
where the sum inside the formula is to be understood as shorthand 
for taking the sum over the bounded number of possible values of $i$
and $j$.  

Now, we can define the relation $V \sse [n]^t \times U^k$ given by
$V(g,\bar a)$ if, and only if, $\bar a \in \tEV{g}$ by the following formula
$$ 
\theta(\mu,\bar x) \de [\ifp_{V,\nu \bar y} \displaystyle \bigvee_{s \in  \B \uplus \tau \uplus \set{0,1}}
(\phi_s(\mu) \land \theta_s(\nu,\bar y))](\mu,\bar x).
$$
The overall $q$-ary output query computed by the circuit is
given by the following formula derived from final the construction in the
last subsection
\begin{equation}
\begin{aligned}
Q(z_1,\ldots,z_q) &\de \exists \mu\bar x \nu_1 \cdots \nu_q \eta_1 \cdots \eta_k
[ \theta(\mu, \bar x) \wedge \phi_\Omega(\nu_1,\ldots,\nu_q, \mu) \wedge \\
& \hspace{23ex} \bigwedge_{1 \le i \le
k} (\supp_i(\mu,\eta_i) \vee \forall \eta [\neg \supp_i(\mu,\eta)]) \wedge \\
& \hspace{23ex} \bigwedge_{1 \le i \le k} \bigwedge_{1 \le j \le q}
([\supp_i(\mu,\eta_i) \wedge x_i = z_j] \rightarrow \nu_j
= \eta_i) \wedge \\
& \hspace{23ex}\bigwedge_{1 \le j \le q} \bigvee_{1 \le i \le k} (x_i =
z_j \wedge \supp_i(\mu,\eta_i))]
\end{aligned}
\end{equation}
where the purpose of the last three lines is to invert the injective
function encoded in $\bar x$ and then apply it to $z_i$ to produce
$\nu_i$; in particular: the second line puts the ordered support of
$\mu$ into $\eta_1,\ldots,\eta_k$, the third line defines the map from
$z_i$ to $\nu_i$, and the fourth line insures that this map covered all
coordinates of $z_i$.

Note that this is a formula of $\FPN$ if $\B$ is the standard basis
and a formula of $\FPC$ if $\B$ is the majority basis.
Moreover, if the family $\C = (C_n)_{n \in \NN}$ of symmetric circuits
is not $\PT$-uniform, but given by an advice function $\Advice$, the
construction gives us an equivalent formula of $\FPG$ (for the
standard basis) or $\FPCG$ (for the majority basis).  This may be
formalised as follows. 
\begin{lemma}\label{lem:main}~
\begin{enumerate}
\item 
Any relational query defined by a $\PT$-uniform family
of circuits over the standard basis is definable in $\FPN$.
\item 
Any relational query defined by a $\PT$-uniform family
of circuits over the majority basis is definable in $\FPC$.
\item 
Any relational query defined by a $\Pp$-uniform family
of circuits over the standard basis is definable in $\FPG$, for some advice function $\Advice$. 
\end{enumerate}
\end{lemma}


\section{Consequences}
\label{sec:conseq}

Formulas of $\FPN$ can be translated into $\PT$-uniform families of
symmetric Boolean circuits by standard methods and similar
translations hold for $\FPC$ and $\FPG$.  This combined with Lemma~\ref{lem:main} proves our main theorem. 
\begin{theorem}[Main]
  \label{thm:main}
  The following pairs of classes define the same queries on
  structures:
  \begin{enumerate}
  \item Symmetric $\PT$-uniform Boolean circuits
    and $\FPN$.
  \item  Symmetric $\PT$-uniform majority
    circuits and $\FPC$.
  \item  Symmetric $\Pp$-uniform majority
    circuits and $\FPCG$.
  \end{enumerate}
\end{theorem}

One consequence is that properties of graphs which we know not to be 
definable in $\FPC$ are also not decidable by $\PT$-uniform
families of symmetric circuits.  The results of
Cai-F\"urer-Immerman \cite{CFI92} give graph properties that are polynomial-time decidable, but not definable in $\FPC$.   Furthermore, there are a
number of $\NP$-complete graph problems known not to be definable in $\FPC$, including Hamiltonicity and
3-colourability (see~\cite{Daw98}).  Indeed, all these proofs actually
show that these properties are not even definable in the infinitary
logic with a bounded number of variables and counting
($C^{\omega}_{\infty\omega}$---see~\cite{Ott97}).  Since it is not
difficult to show that formulas of $\FPCG$ can be translated into
$C^{\omega}_{\infty\omega}$, we have the following.
\begin{corollary}\label{cor:ham}
Hamiltonicity and 3-colourability of graphs are not decidable by
families of $\Pp$-uniform symmetric majority circuits.
\end{corollary}

\section{Coherent and Locally Polynomial Circuits}
\label{sec:coherence}

In this section we discuss connections with the prior work of Otto
\cite{O97}. Otto studies rigid symmetric Boolean circuits
deciding Boolean properties of structures and gives uniformity
conditions on such families that characterise bounded-variable
fragments of finite and infinitary first-order logic.  Otto defines
two properties to establish his notion of uniformity.  The first
property is called coherence; informally, a circuit family
$(C_n)_{n \in \NN}$ is coherent if $C_n$ appears as a subcircuit of
all but finitely many of the circuits at larger input lengths.

\begin{definition}[Coherence]
  Let $\C \de (C_n)_{n \in
    \NN}$ be a family of rigid symmetric $(\Bs,\tau)$-circuits
  computing a Boolean function.  The circuit $C_n$ \emph{embeds} into
  the circuit $C_m$ with $m > n$ if there is a subcircuit of $C_m$
  which is isomorphic to $C_n$.  An embedding is \emph{complete} if
  its images are exactly those gates of $C_m$ which are fixed by
  $\Sym{[m] \bs [n]}$.  The circuit family $\C$ is \emph{coherent} if
  for each $n \in \NN$, $C_n$ completely embeds into $C_m$ for all
  large enough $m > n$. 
\end{definition}

The second property is locally polynomial; informally, a circuit
family is locally polynomial if the size of the orbit of every wire is
polynomially bounded.

\begin{definition}[Locally Polynomial]
  A rigid circuit family $(C_n)_{n \in \NN}$ is \emph{locally polynomial of
    degree} $k$ if there is a $k \in \NN$ such that each $C_n$ and
  every subset $S \sse [n]$, the size of the orbit of every wire with
  respect to the automorphisms of the circuit induced by $\Sym{S}$ at
  most $|S|^k$.
\end{definition}

The main result of \cite[Theorem 6]{O97} establishes an equivalence
between  coherent
locally-polynomial (of degree $k$) families of rigid symmetric
$(\Bs,\tau)$-circuits computing Boolean functions on $\fin[\tau]$
and infinitary first-order logic with $k$ variables.
It should be noted that in Otto's definition of circuit families the individual circuits in the family may
themselves be infinite, as the only size restriction is on the orbits
of gates.
The theorem also
shows that if the circuit families are also constant depth they
correspond to the fragment of first-order logic with $k$ variables.


The common restriction of notions of uniformity we consider in this
paper is that the circuits have size polynomial in their input length.
If we restrict ourselves to  locally-polynomial coherent symmetric
families where the individual circuits are \emph{finite}, 
we can use the Support Theorem (Corollary~\ref{cor:support}) to
establish a direct connection with polynomial-size symmetric
circuit families, formally stated in the following
proposition. 
  
\begin{prop}\label{prop:coherent}
  Let $\C \de (C_n)_{n \in \NN}$ be a family of rigid symmetric
  Boolean circuits.
  \begin{enumerate}
  \item If $\C$ is a locally-polynomial coherent family, then $\C$ is
    polynomial size.
  \item If $\C$ is polynomial size, then $\C$ is locally polynomial.
  \end{enumerate}
\end{prop}

\begin{proof}
  We prove the two parts separately.
  
  \iparagraph{Part 1.}  Suppose to the contrary that $C_n$ has $s(n) =
  \omega(\poly(n))$ gates.  Because $\C$ is locally polynomial the
  Support Theorem gives a bound $k \in \NN$ on the size of the support
  of gates in $\C$.  Take $m \in \NN$ such that $C_m$ is a circuit
  such that $C_k$ completely embeds into $C_m$ and $s(k) \cdot m^k <
  s(m)$, such $m$ exists because $\C$ is coherent and $s$ is super
  polynomial.  By symmetry and averaging there are at least $\frac{s(m)}{m^k}$
  gates of $C_m$ whose supports are drawn from $[k]$.  These gates are
  necessarily fixed by $\Sym{[m] \bs [k]}$.  Since the embedding is
  complete, $C_k$ maps onto at least these gates.  But this is a
  contradiction because $s(k) < \frac{s(m)}{m^k}$.  Thus $\C$ has
  polynomially many gates.
  
  \iparagraph{Part 2.}  If $\C$ has polynomially many gates then the
  Support Theorem immediately implies that the supports of all gates
  in $\C$ is bounded by some $k \in \NN$.  Therefore for every $S \sse
  [n]$ and every wire in $C_n \in \C$ has its orbit size bounded by
  $|S|^{2k}$.  This is exactly the definition of locally polynomial.
\end{proof}
Since there are properties definable in an infinitary logic with
finitely many variables that are not decidable by polynomial-size
circuits, it follows from the above proposition that the use of
infinite circuits is essential in Otto's result.

Proposition~\ref{prop:coherent} implies that all uniform circuit families we
consider are locally polynomial.  However, it does not establish an
equivalence between a circuit family having polynomially many gates
and being locally polynomial and coherent.  Indeed there are
Boolean circuit families uniformly definable in \FON that are not
coherent.  To see this observe that such circuit families may
include gates that are completely indexed by the number sort and
hence are fixed under all automorphisms induced by permutations of the
point sort.  Moreover the number of such gates may increase as a
function of input length.  However, because coherence requires that
\emph{complete} embedding exist, the number of gates in each circuit
of a coherent family that are not moved by any automorphism must be
identical.  Thus there are uniform circuits that are not
coherent.

Consider weakening the definition of coherence to require only that an embedding
exists but not that the embedding is complete, and call this
\emph{partial coherence}.  One can show that any relation which can be
computed by a Boolean circuit family uniformly definable in \FON can also be computed by a partially coherent Boolean circuit
family with the same uniformity by appropriately creating copies of circuits relativised for
all shorter lengths.  
We omit any formal discussion of this construction.

\section{Future Directions}
\label{sec:conc}
One of the original motivations for studying symmetric majority
circuits was the hope that they had the power of choiceless polynomial
time with counting (\CPTC) \cite{BGS99}, and that, perhaps, techniques
from circuit complexity could improve our understanding of the
relationship between \CPTC and the invariant queries definable in
polynomial-time.  However, because $\FPC \ssn \CPTC$ \cite{DRR08}, our
results indicate that symmetry is too much of restriction on
$\PT$-uniform circuit families to recover $\CPTC$.

A natural way to weaken the concept of symmetry is to require only
that induced automorphisms exist only for a certain subgroup of the
symmetric group.  This interpolates between our notion
symmetric circuits and circuits on linearly-ordered structures, with
the latter case occurring when the subgroup is the identity.  An
easier first step may be to consider the action on structures with a
finite number of disjoint sorts and require only that automorphisms be
induced by permutations which preserve the sorts, e.g., structures
interpreting Boolean matrices whose rows and columns are indexed by disjoint
sets.

The Support Theorem is a fairly general statement about the structure
of symmetric circuits and is largely agnostic to the particular
semantics of the basis.  To that end the Support Theorem may find
application to circuits over bases not consider here.  The Support
Theorem can be applied to arithmetic circuits computing invariant
properties of matrices over a field; e.g., the Permanent polynomial is
invariant and one standard way to compute it is as a symmetric
arithmetic circuit, i.e., Ryser's formula \cite{R63}.  Finally, the form of
the Support Theorem can, perhaps, be improved as the particular upper
bound required on the orbit size does not appear to be fundamental to
the conclusion of the Support Theorem.

\paragraph{Acknowledgments.} The authors thank Dieter van Melkebeek for looking at an early draft of this paper.  This research was supported by EPSRC grant EP/H026835.

\bibliographystyle{plain}  
\bibliography{main}

\begin{thebibliography}{10}

\bibitem{BGS99}
A.~Blass, Y.~Gurevich, and S.~Shelah.
\newblock Choiceless polynomial time.
\newblock {\em Ann. Pure Appl. Logic}, 100:141--187, 1999.

\bibitem{CFI92}
J-Y. Cai, M.~F\"{u}rer, and N.~Immerman.
\newblock An optimal lower bound on the number of variables for graph
  identification.
\newblock {\em Combinatorica}, 12(4):389--410, 1992.

\bibitem{Daw98}
A.~Dawar.
\newblock A restricted second order logic for finite structures.
\newblock {\em Information and Computation}, 143:154--174, 1998.

\bibitem{DRR08}
A.~Dawar, D.~Richerby, and B.~Rossman.
\newblock Choiceless polynomial time, counting and the
  {C}ai--{F}{\"u}rer--{I}mmerman graphs.
\newblock {\em Ann. Pure Appl. Logic}, 152(1):31--50, 2008.

\bibitem{DGS86}
L.~Denenberg, Y.~Gurevich, and S.~Shelah.
\newblock Definability by constant-depth polynomial-size circuits.
\newblock {\em Inform. Control}, 70(2):216--240, 1986.

\bibitem{EF06}
H.D. Ebbinghaus and J.~Flum.
\newblock {\em Finite Model Theory}.
\newblock Springer, 2006.

\bibitem{I86}
N.~Immerman.
\newblock Relational queries computable in polynomial time.
\newblock {\em Inform. Control}, 68(1-3):86--104, 1986.

\bibitem{Ott97}
M.~Otto.
\newblock {\em Bounded Variable Logics and Counting: A Study in Finite Models},
  volume~9 of {\em Lecture Notes in Logic}.
\newblock Springer-Verlag, 1997.

\bibitem{O97}
M.~Otto.
\newblock The logic of explicitly presentation-invariant circuits.
\newblock In Dirk Dalen and Marc Bezem, editors, {\em Computer Science Logic},
  volume 1258 of {\em Lecture Notes in Computer Science}, pages 369--384.
  Springer Berlin Heidelberg, 1997.

\bibitem{R63}
H.J. Ryser.
\newblock {\em {Combinatorial Mathematics}}.
\newblock Mathematical Association of America, 1963.

\bibitem{V82}
M.~Vardi.
\newblock The complexity of relational query languages.
\newblock In {\em Proceedings of the Fourteenth Annual ACM Symposium on Theory
  of Computing}, pages 137--146. ACM, 1982.

\bibitem{V99}
H.~Vollmer.
\newblock {\em {Introduction to Circuit Complexity: A Uniform Approach}}.
\newblock Springer, 1999.

\end{thebibliography}








\end{document}